\documentclass[letterpaper,journal,twoside]{IEEEtran}
\usepackage{url}
\usepackage[sort,nocompress]{cite}
\usepackage{amsmath,amssymb,amsfonts}
\usepackage[table,xcdraw]{xcolor}
\usepackage[ruled,linesnumbered,vlined]{algorithm2e}
\usepackage{array}
\usepackage{textcomp}
\usepackage{stfloats}
\usepackage{verbatim}
\usepackage{graphicx}
\usepackage{bbm}
\usepackage{multirow}
\usepackage{leftindex}
\usepackage{latexsym}
\usepackage{mathrsfs}
\usepackage{setspace}
\usepackage{booktabs}
\usepackage{color}
\usepackage{comment}
\usepackage[font=footnotesize,labelsep=period]{caption}
\usepackage{makecell}
\usepackage{threeparttable}
\usepackage{colortbl}
\usepackage{float}
\usepackage{lipsum}
\usepackage{bm}
\usepackage{balance}
\usepackage{cases}
\newtheorem{theorem}{Theorem}

\newtheorem{assumption}{Assumption}
\newtheorem{remark}{Remark}
\newtheorem{proposition}{Proposition}
\makeatletter
\renewcommand{\maketag@@@}[1]{\hbox{\m@th\normalsize\normalfont#1}}%
\makeatother
\usepackage[subrefformat=parens,justification=centering]{subcaption}
\newcolumntype{P}[1]{>{\centering\arraybackslash}m{#1}}
\usepackage{enumitem}

\usepackage{CJKutf8}

\newcommand{\citecolor}{blue}
\usepackage[colorlinks=true,allcolors=\citecolor,urlcolor=black]{hyperref}

\let\oldeqref\eqref
\renewcommand{\eqref}[1]{\textcolor{\citecolor}{\oldeqref{#1}}}

\usepackage{orcidlink}
\newcommand{\IEEEorcidlink}[1]{\raisebox{0.5mm}{\kern0.15mm\scalebox{1.2}{\orcidlink{#1}}}}

\newcommand{\email}[1]{\mbox{\href{mailto:#1}{#1}}}

\newcommand{\T}{\mathsf{T}}
\newcommand{\m}{\,\mathrm{m}}
\newcommand{\cm}{\,\mathrm{cm}}

\newcommand{\LED}{\text{LED}}
\newcommand{\pnp}{P\textit{n}P}
\newcommand{\freepnp}{Free\pnp}
\newcommand{\wrt}{w.r.t.}
\usepackage{tikz}
\newcommand{\circLED}{\tikz\draw[line width=1pt](0,0)circle(3.3pt);}
\newcommand{\rectLED}{\tikz\draw[line width=1pt](0,0)rectangle(6pt,6pt);}

\newcommand{\seqref}[1]{\scalebox{0.9}{\eqref{#1}}}

\begin{document}

\title{
Visible Light Positioning With Lam\'e Curve LEDs: A Generic Approach for Camera Pose Estimation
}

\author{%
    Wenxuan~Pan\IEEEorcidlink{0009-0000-7922-1583},~\IEEEmembership{Graduate~Student~Member,~IEEE}, Yang~Yang\IEEEorcidlink{0000-0002-1353-0890},~\IEEEmembership{Senior~Member,~IEEE}, Dong~Wei\IEEEorcidlink{0009-0001-6508-4490}, Zhiyu~Zhu\IEEEorcidlink{0000-0001-8548-5825}, Jintao~Wang\IEEEorcidlink{0009-0002-3335-633X}, Huan~Wu\IEEEorcidlink{0000-0003-1393-7049}, and~Yao~Nie\IEEEorcidlink{0000-0002-0716-6404}%
    \thanks{Received 2 February 2026. This work was supported in part by the National Natural Science Foundation of China under Grant 62371065 and Grant 61871047, and in part by BUPT Innovation and Entrepreneurship Support Program under Grant 2025-YC-S007. \textit{(Corresponding author: Yang~Yang.)}}%
    \thanks{Wenxuan~Pan and Yang~Yang are with Beijing Key Laboratory of Network System Architecture and Convergence, School of Information and Communication Engineering, Beijing University of Posts and Telecommunications, Beijing 100876, China (e-mail: \email{pwx@bupt.edu.cn}; \email{yangyang01@bupt.edu.cn}).}%
    \thanks{Dong~Wei is with the Institute of Information Engineering, Chinese Academy of Sciences, Beijing 100085, China (e-mail: \email{weidong@iie.ac.cn}).}%
    \thanks{Zhiyu~Zhu is with the College of Physics and Electronic Engineering, Shanxi University, Taiyuan 030006, China (e-mail: \email{zhiyu.zhu@sxu.edu.cn}).}%
    \thanks{Jintao~Wang, Huan~Wu, and Yao~Nie are with the School of Electronic Information and Artificial Intelligence, West Anhui University, Lu'an 237012, China (e-mail: \email{2025210028@mail.wxc.edu.cn}; \email{wuhuan@wxc.edu.cn}; \email{nieyao@wxc.edu.cn}).}
    \thanks{Color versions of one or more figures in this article are available at xxx.}%
    \thanks{Digital Object Identifier xxx}%
    \vspace{-0.5cm}%
}

\markboth{SUBMITTED TO AN IEEE JOURNAL FOR POSSIBLE PUBLICATION,~Vol.~X, No.~X, FEBRUARY~2026}%
{Pan \MakeLowercase{et al.}: A Sample Article Using IEEEtran.cls for IEEE Journals}

\IEEEpubid{
    \parbox{\textwidth}{\centering 0000-0000~\copyright~2026 IEEE.}
}

\maketitle

\begin{abstract}
Camera-based visible light positioning (VLP) is a promising technique for accurate and low-cost indoor camera pose estimation (CPE). To reduce the number of required light-emitting diodes (LEDs), advanced methods commonly exploit LED shape features for positioning. Although interesting, they are typically restricted to a single LED geometry, leading to failure in heterogeneous LED-shape scenarios. To address this challenge, this paper investigates Lam\'e curves as a unified representation of common LED shapes and proposes a generic VLP algorithm using Lam\'e curve-shaped LEDs, termed LC-VLP. 
In the considered system, multiple ceiling-mounted Lam\'e curve-shaped LEDs periodically broadcast their curve parameters via visible light communication, which are captured by a camera-equipped receiver. Based on the received LED images and curve parameters, the receiver can estimate the camera pose using LC-VLP. 
Specifically, an LED database is constructed offline to store the curve parameters, while online positioning is formulated as a nonlinear least-squares problem and solved iteratively. To provide a reliable initialization, a correspondence-free perspective-\textit{n}-points (FreeP\textit{n}P) algorithm is further developed, enabling approximate CPE without any pre-calibrated reference points. 
The performance of LC-VLP is verified by both simulations and experiments. {Simulations show that LC-VLP outperforms state-of-the-art methods in both circular- and rectangular-LED scenarios. Compared to a perspective arcs algorithm, LC-VLP can achieve reductions of both over 30\% in average position and rotation errors.} Experiments further show that LC-VLP can achieve an average position accuracy of less than 4 cm.
\end{abstract}


\section{Introduction}
\IEEEPARstart{I}{ndoor} positioning has become a fundamental enabling technology for a wide range of location-aware applications in smart buildings, robotics, and the Internet of Things (IoT). Although the global navigation satellite systems (GNSSs) are predominant and well-established solutions that can provide reliable positioning performance in outdoor scenarios, they face challenges indoors, including severe signal attenuation and degraded reliability\cite{yang2024,yao2024}. Under this background, a variety of indoor positioning technologies have been developed, such as WiFi\cite{yu2024}, Bluetooth low energy (BLE)\cite{pan2025}, ultra-wideband (UWB)\cite{wang2025}, radio frequency identification (RFID)\cite{xu2023}, and visible light positioning (VLP)\cite{Xu2025,Shi2025,SU2023,vpca,vovlp,He2024VLP,pan2025VLP}. Among these, the VLP technology has attracted increasing attention, as it can achieve high accuracy while leveraging existing lighting infrastructure at a low cost\cite{Zhu2024ASurvey}.

\IEEEpubidadjcol

{Depending on whether optical modulation is required, VLP can generally be categorized into modulated and unmodulated systems\cite{Alijani2025}. Moreover, according to whether the target carries a receiving device, VLP can also be classified as device-based and device-free systems\cite{Wang2020Passive,Alijani2026}. This work focus on the modulated, device-based VLP, which employs light-emitting diodes (LEDs) as transmitters, periodically broadcasting data packets containing essential information, such as the LED identities (IDs), via visible light communication (VLC).} To receive these signals, two types of receivers are commonly used: Photodiodes (PDs)\cite{Shi2025,Xu2025,SU2023} and cameras\cite{vpca,vovlp,He2024VLP,pan2025VLP}, each with its own advantages and limitations\cite{Zhu2024ASurvey,Bastiaens2024,Li2025}. Camera-based methods typically utilize complementary metal-oxide-semiconductor (CMOS) cameras to capture images of LED luminaires. By exploiting the rolling-shutter effect (RSE), the LED IDs can be detected and decoded, while visual features extracted from the captured images can be jointly used for 6-degrees of freedom (DoF) camera pose estimation (CPE). Besides their strong robustness to ambient light interference\cite{vpca,vovlp}, camera-based methods leverage the ubiquitous availability of cameras which support a wide range of applications beyond positioning. Consequently, camera-based methods are regarded as practical, flexible, and easy to integrate into existing systems\cite{Zhu2024ASurvey,Bastiaens2024}. {Nevertheless, camera-based VLP may still face several practical challenges, including relatively high computational complexity, increased power consumption, and potential privacy concerns associated with image acquisition and processing\cite{Saeed2019}.} However, achieving high-accuracy VLP with fewer LEDs while maintaining robustness across heterogeneous LED shapes remains a challenging open problem.

\subsection{Related Works}
Existing camera-based VLP methods can generally be categorized into two classes: Point-source-based methods\cite{Kuo2014,isasvg} and planar-source-based methods\cite{zhang2017,vpca,vovlp,cheng2020,liu2022,vp4l,vlcpnp,Yang2022}. Point-source-based methods exploit only the centroids of LED projections and its positioning process can be formulated as a perspective-\textit{n}-points (\pnp) problem\cite{zhang2023,epnp,opnp} when multiple LEDs are visible. However, the requirement of observing $\geq4$ LEDs imposes strict constraints on LED deployment density\cite{Bai2020}. Although auxiliary sensors such as inertial measurement units (IMUs) can be introduced to reduce the number of required LEDs\cite{isasvg}, they inevitably increase system complexity and hardware cost.

In contrast, planar-source-based methods explicitly exploit the geometric shape of LEDs for positioning. By incorporating fine-grained shape information, these methods typically require fewer visible LEDs, and have therefore become the dominant paradigm. However, existing studies are usually restricted to a single specific LED shape, such as circular LEDs\cite{zhang2017,vpca,vovlp,cheng2020,liu2022} or rectangular LEDs\cite{vp4l,vlcpnp,Yang2022}. The works in\cite{vovlp,vpca} and\cite{zhang2017} proposed VLP approaches that leverage the geometric characteristics of circular LEDs without other auxiliary sensors. Specifically, Zhang et al.\cite{zhang2017} adopted a weak-perspective projection model which guarantees solution uniqueness at the cost of reduced accuracy. To overcome this limitation, Zhu et al.\cite{vpca} employed the pinhole projection model and resolved the inherent ambiguity by exploiting an additional visible LED. They further introduced the visual odometry (VO)\cite{vovlp} to enable positioning with a single circular LED under dynamic camera motion scenarios. For rectangular LEDs, Bai et al.\cite{vp4l} proposed using the perspective-4-line (P4L) algorithm for CPE. However, this approach requires VLC information to be transmitted at a specific LED corner, which is difficult to implement on commercial off-the-shelf rectangular luminaires\cite{Narasimman2024}. Moreover, the works in\cite{cheng2020,liu2022,vlcpnp} and\cite{Yang2022} introduced additional IMUs or magnetometers to assist in CPE. As discussed earlier, such auxiliary sensors not only increase hardware complexity but also suffer from degraded reliability under electromagnetic interferences, thereby limiting their practical applicability. To summarize, although the above works\cite{zhang2017,vpca,vovlp,cheng2020,liu2022,vp4l,vlcpnp,Yang2022} exploit LED shape feature for positioning, they are restricted to a single specific LED geometry, and thus they fail in heterogeneous LED-shape scenarios or when other common LED shapes, e.g., rhombic or elliptical LEDs, are deployed. Moreover, these methods rely on purely geometric solutions, which are sensitive to image noises since the camera pose is not refined through global iterative optimization. 

Motivated by the above limitations, a unified representation that can characterize diverse planar LED shapes is highly desirable. \textit{Lam\'e curves}\cite{Kajikawa2025,NI2016,Masaya2015}, also referred to as the \textit{superellipses}, constitute a general family of parametric curves that naturally subsume common LED shapes as special cases. {In practical indoor environments, commercial LED luminaires are manufactured in various planar geometries, including circular, rectangular, rhombic, and elliptical shapes, which can be naturally modeled by Lam\'e curves.} Therefore, introducing Lam\'e curves into VLP can provide a unified framework that accommodates multiple LED shapes within a single algorithm. Nevertheless, this generality comes at the cost of more complex curve equations, which renders conventional purely geometric methods inapplicable. Therefore, exploiting Lam\'e curve LEDs for VLP deserves further investigation.

\subsection{Contributions}
The primary contribution of this paper is a generic VLP algorithm based on Lam\'e curve-shaped LEDs, termed LC-VLP. \textit{To the authors' best knowledge, this is the first generic VLP approach that unifies the modeling of various commonly-used LED shapes and achieves full 6-DoF CPE.} Our key contributions are summarized as follows:
\begin{itemize}
    \item We investigate a unified representation for commonly-used LED shapes, namely Lam\'e curves. {For the first time, Lam\'e curves are introduced into VLP to provide a unified parametric modeling that naturally covers a wide range of practical planar LED geometries, including circular, rectangular, rhombic, and elliptical shapes commonly found in commercial luminaires. To accommodate shape diversity in realistic heterogeneous LED deployments, an LED database construction method is further developed to associate each LED with its corresponding Lam\'e-curve parameters, enabling a single unified framework without algorithm switching.}
    \item Based on the Lam\'e-curve modeling of LEDs with different shapes, we propose a generic LC-VLP algorithm that achieves higher accuracy and robustness in heterogeneous and mixed LED deployments. Specifically, LC-VLP is built upon a \textit{back-projection strategy}, where the LED projection points are back-projected onto the ceiling plane. A nonlinear least-squares (NLLS) optimization model is then formulated, where iterative refinement is performed to jointly exploit all observed LED contours for accurate 6-DoF CPE.
    \item The NLLS optimization requires a reliable initial estimate to ensure convergence. To this end, we propose a correspondence-free \pnp\ (\freepnp) algorithm. Unlike conventional \pnp\ methods, \freepnp\ does not require prior knowledge of pre-calibrated 3D--2D reference point (RP) correspondences, yet it is still capable of coarse CPE. Specifically, we theoretically prove a \textit{collinearity invariance theorem}. Based on this theorem and an assumption of uniform distribution of LED contour points, multiple approximate virtual correspondences can be constructed. The camera pose is initialized by applying a \pnp\ solver to these virtual correspondences.
\end{itemize}

{Substantial simulation results show that LC-VLP outperforms state-of-the-art (SoTA) methods in both circular- and rectangular-LED scenarios. Compared to a perspective arcs algorithm\cite{vpca} and a VLC-based \pnp\ algorithm\cite{vlcpnp}, LC-VLP can achieve reductions of both over 30\% in average position and rotation errors.} We further implement LC-VLP by developing an experimental prototype in a practical scenario. Based on this prototype, experiments are implemented, which show that LC-VLP can achieve an average position accuracy of less than 4 cm.

\subsection{Notations and Paper Organization}
In this paper, matrices and column vectors are denoted by uppercase and lowercase boldface, e.g., $\bm{R}$ and $\bm{v}$, respectively, with $(\cdot)^\T$ denoting their transposition. For vector $\bm{v}$, $\|\bm{v}\|$ represents its Euclidean norm, and $\widetilde{\bm{v}}\triangleq[\bm{v}^\T,1]^\T$ denotes its homogeneous representation. For square matrix $\bm{R}$, $\det(\bm{R})$ and $\mathrm{tr}(\bm{R})$ denote its determinant and trace, respectively. Sets are denoted by calligraphic uppercase letters, e.g., $\mathcal{A}$, with $|\mathcal{A}|$ indicating its cardinality, and $\{a_i\}_{i=1}^N$ denoting a sequence. Points are represented by uppercase letters such as $P$, and $\overline{P_1P_2}$ denotes the segment connecting points $P_1$ and $P_2$.

The rest of the paper is organized as follows. Section~\ref{sec:syst} introduces the system model. In Section~\ref{sec:lcvlp}, the proposed \freepnp\ and LC-VLP algorithms are presented in detail. Then, Section~\ref{sec:simexp} presents the simulation and experimental setup and results. Finally, Section~\ref{sec:con} concludes this paper.

\section{System Model and Problem Formulation}\label{sec:syst}

\subsection{System Model}

\begin{figure}[!t]
    \centering
    \includegraphics[width=0.41\textwidth]{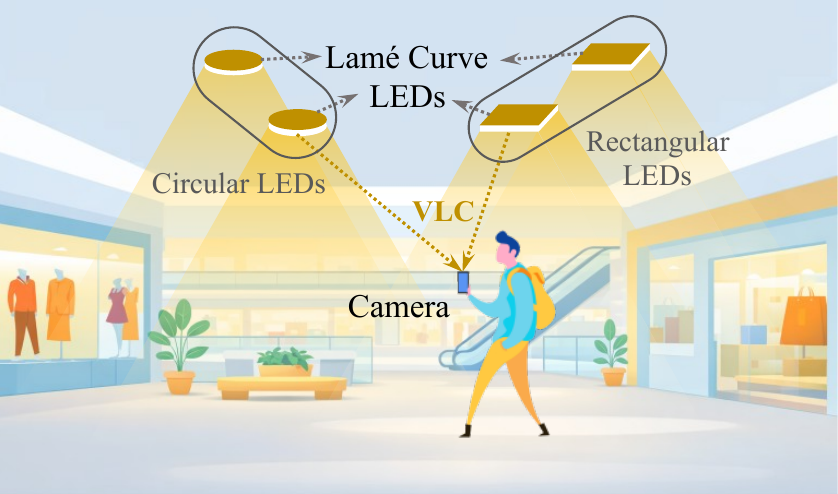}
    \caption{{The considered VLP scenario, where the circular and rectangular LEDs on the ceiling are uniformly modeled as Lam\'e curve-shaped LEDs.}}
    \label{fig:sce}
\end{figure}

The considered typical indoor VLP scenario is illustrated in Fig.~\ref{fig:sce}. {In this scenario, multiple modulated Lam\'e curve-shaped LEDs are installed in the ceiling. These LEDs function as transmitters that continuously broadcast VLC information while providing illumination.} In practical environments, such LEDs often exhibit heterogeneous shapes due to mixed lighting infrastructures, e.g., circular downlights combined with rectangular panel lights in offices, or the coexistence of legacy luminaires and newly retrofitted LED fixtures in shopping malls and transportation hubs. The positioning device is equipped with a camera capable of simultaneously receiving both VLC signals and visual information from the LEDs, thereby enabling real-time, onboard estimation of its own 6-DoF pose. In particular:

\subsubsection{Lam\'e Curve-Shaped LEDs}
Lam\'e curves\cite{Kajikawa2025,NI2016,Masaya2015} are commonly used as design elements, particularly in the shapes of LED luminaires. On the plane $xOy$, consider a Lam\'e curve with center $(x_0,y_0)^\T$. Its major and minor axes are aligned with the $x$- and $y$-axes, with lengths $2a$ and $2b$, respectively, and its order is denoted by $\gamma$. Then, the equation of this curve is given by: 
\begin{equation}\label{eqn:lame}
    \big|(x-x_0)/a\big|^{\gamma} + \big|(y-y_0)/b\big|^{\gamma} = 1.
\end{equation}
The order $\gamma$ determines the geometric shape of this curve: 
\begin{itemize}
    \item When $\gamma=1$, the curve becomes a rhombus, and in the special case $a = b$, it reduces to a square (with vertices lying on the coordinate axes).
    \item When $\gamma=2$, the curve becomes an ellipse, and in the special case $a=b$, it reduces to a circle.
    \item As $\gamma\rightarrow+\infty$, the curve approaches a rectangle, and in the special case $a=b$, it reduces to a square. In practice, a sufficiently large value of $\gamma$, such as $100$ or larger, can be used to obtain an accurate approximation.
\end{itemize}

The Lam\'e curve formulation in \eqref{eqn:lame} can provide a unified representation for common LED shapes such as circles\cite{vpca} and rectangles\cite{vp4l}. This unification establishes the theoretical foundation for the proposed generic LC-VLP algorithm.

\begin{figure}[!t]
    \centering
    \includegraphics[width=0.485\textwidth]{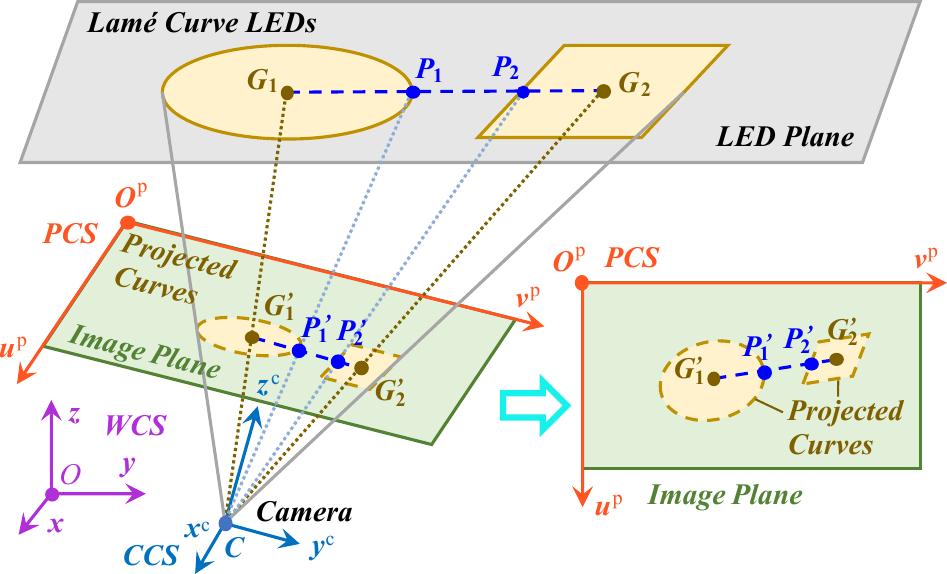}
    \caption{The system projection diagram.}
    \label{fig:syst}
\end{figure}

\subsubsection{Projection Model}
As shown in Fig.~\ref{fig:syst}, we adopt the pinhole projection model\cite{cvref} to describe the camera imaging process. First, we establish the following coordinate systems: (i) the 3D world coordinate system (WCS) $O$--$xyz$, (ii) the 3D camera coordinate system (CCS) $C$--$x^\mathrm{c}y^\mathrm{c}z^\mathrm{c}$, and (iii) the 2D pixel coordinate system (PCS) $u^\mathrm{p}O^\mathrm{p}v^\mathrm{p}$. Among these, in WCS, the $z$-axis points vertically upward toward the ceiling. The origin of CCS $C$ coincides with the camera’s optical center; Moreover, to ensure visibility of the LEDs, the $z^{\mathrm{c}}$-axis is oriented approximately upward and is perpendicular to the image plane. The PCS lies on the image plane, and in 3D space, its $u^{\mathrm{p}}$- and $v^{\mathrm{p}}$-axes are parallel to the $x^\mathrm{c}$- and $y^\mathrm{c}$-axes of CCS, respectively. In the following, the coordinates of a point expressed in the WCS/CCS/PCS are referred to as its world/camera/pixel coordinates.

Consider a 3D point $\bm{x}^\mathrm{c}$ in CCS, whose projection onto the image plane yields a 2D image point $\bm{u}^\mathrm{p}$. This projection process can be modeled as:
\begin{equation}
\label{eqn:c2p}
    z^\mathrm{c}\widetilde{\bm{u}}^\mathrm{p} = \bm{K}\bm{x}^\mathrm{c},
\end{equation}
where $\bm{K}\triangleq\left[\begin{smallmatrix}f_x&0&u_0\\0&f_y&v_0\\0&0&1\end{smallmatrix}\right]$ is the intrinsic matrix of camera $C$ which can be calibrated in advance. In matrix $\bm{K}$, $f_x$ and $f_y$ are the focal lengths in pixels; and $(u_0, v_0)^\T$ represent the pixel coordinates of the principal point. 

We further consider the transformation between CCS and WCS. Since the two coordinate systems share the same scale, the relationship between the world coordinates $\bm{x}$ of a point in 3D space and its camera coordinates $\bm{x}^\mathrm{c}$ is given by:
\begin{equation}
\label{eqn:w2c}
    \bm{x}^\mathrm{c} = \bm{R}\bm{x}+\bm{t},
\end{equation}
where $\bm{R}\in\mathfrak{SO}(3)$ is the rotation matrix representing the camera’s orientation, $\mathfrak{SO}(3)$ represents the 3D special orthogonal group, and $\bm{t}\in\mathbb{R}^3$ is the translation vector representing the camera’s position offset. Matrix $\bm{R}$ and vector $\bm{t}$ constitute the camera extrinsics.

\subsection{Problem Formulation}

The camera's 6-DoF pose includes its world coordinates $\bm{c}$ and its 3D rotation matrix $\bm{R}$. The objective of this work is to estimate this 6-DoF pose, given the pre-calibrated intrinsic matrix $\bm{K}$, the world-coordinate Lam\'e curve equations $\mathcal{E}_i$ of the $N$ observed LEDs (for $i=1,\ldots,N$), and the corresponding sets of pixel-coordinate projection points $\mathcal{I}_i$ on the image plane. Moreover, if additional pre-calibrated RPs are available, these RPs can be incorporated to further assist the positioning process. We denote the set of world coordinates of these RPs and the set of their corresponding pixel-coordinate projection points as $\mathcal{M}$ and $\mathcal{M}'$, respectively.

Based on the aforementioned system model, the camera optical center coincides with the origin of CCS. In the field of computer vision (CV), it is common to use the optical center as the representation of the camera\cite{cvref}; thus, the camera coordinate of camera $C$ is simply $\mathbf{0}$. Substituting $\bm{c}^{\mathrm{c}}=\mathbf{0}$ into \eqref{eqn:w2c} yields:
\begin{equation}
\label{eqn:tar}
    \bm{c}=-\bm{R}^\T\bm{t},
\end{equation}
from which we can see that estimating the world coordinates of camera $C$ is equivalent to solving for the extrinsics $\bm{R}$ and $\bm{t}$. Accordingly, we formulate the positioning process as: 
\begin{subequations}
\begin{gather}
    \bm{R},\bm{t}=\bm{f}\big(\{\mathcal{I}_i\}_{i=1}^N,\mathcal{M'};\{\mathcal{E}_i\}_{i=1}^N,\mathcal{M}\big) \tag{\theequation}\label{eqn:prob1}\\
    \mathrm{s.t.}~~\bm{R}\in\mathfrak{SO}(3),\;\bm{t}\in\mathbb{R}^{3},\;-\bm{R}^\T\bm{t}\in\mathcal{U}. \label{eqn:prob1-1}
\end{gather}
\end{subequations}
In \eqref{eqn:prob1}, the function $\bm{f}(\cdot)$ denotes the mapping used to compute the extrinsics. In \eqref{eqn:prob1-1}, the constraints $\bm{R}\in\mathfrak{SO}(3)$ and $\bm{t}\in\mathbb{R}^{3}$ specify the admissible forms of the rotation matrix and translation vector, respectively, while the constraint $-\bm{R}^\T\bm{t}\in\mathcal{U}$ ensures that the estimated camera position lies within a predefined feasible region $\mathcal{U}$, e.g., not above the ceiling or outside the indoor environment, etc.

To address problem \eqref{eqn:prob1}, we propose an LC-VLP algorithm that enables a generic and unified solution for camera position and orientation estimation with Lam\'e curve-shaped LEDs of arbitrary order. The details of the proposed algorithm are presented in the following sections.

\section{Generic VLP With Lam\'e Curve LEDs}\label{sec:lcvlp}

This section presents the proposed LC-VLP method. First, a database of Lam\'e curve-shaped LEDs is constructed, and a local auxiliary coordinate system is defined for each LED. Then, leveraging the proposed projection collinearity invariance theorem, a \freepnp\ algorithm is developed to provide an initial estimate of the camera 6-DoF pose without requiring pre-calibrated 3D--2D RP correspondences. Finally, a back-projection strategy is adopted to formulate an NLLS optimization model, which is solved using nonlinear optimization methods. The computational complexity of LC-VLP is also analyzed. Detailed descriptions are provided below.

\subsection{LED Database Construction}\label{sec:leddata}

In the considered scenario, the positions, shapes, and sizes of the Lam\'e curve-shaped LEDs may differ from each other. Consequently, an LED database should be constructed in advance to store these parameters. When positioning, the parameters of the captured LEDs are received by the CMOS camera via VLC and RSE\cite{vpca,vovlp}. 

In WCS, assume there are $N'$ LEDs on the ceiling plane, whose equation is $z=z_0$. The boundary curve of the $i$-th LED is denoted by $\mathcal{E}_i$, $1\leq i\leq N'$, and its center is $\overline{\bm{x}}_i=(\overline{x}_i,\overline{y}_i,z_0)^\T$. Then, the equation of $\mathcal{E}_i$ is given by:
\begin{subequations}\label{eqn:ledlame}
\begin{numcases}{}
    \big|(x-\overline{x}_i)/{a_i}\big|^{\gamma_i} + \big|(y-\overline{y}_i)/{b_i}\big|^{\gamma_i} = 1,\label{eqn:ledlame1}\\
    z=z_0.
\end{numcases}
\end{subequations}

To facilitate subsequent point sampling on the curve, we transform \eqref{eqn:ledlame} into a cylindrical coordinate system (CyCS) for auxiliary processing. The origin, $x$- and $z$-axes of this CyCS coincide with those of WCS, and its scale is also identical to that of WCS. Under this transformation, the equation of LED curve $\mathcal{E}_i$ becomes:
\begin{equation}
    \label{eqn:polar}
    \begin{bmatrix}x\\y\\z\end{bmatrix} = \begin{bmatrix}\overline{x}_i\\\overline{y}_i\\z_0\end{bmatrix} + \rho_i(\theta) \begin{bmatrix}\cos\theta\\\sin\theta\\0\end{bmatrix} \triangleq \overline{\bm{x}}_i+\rho_i(\theta)\cdot\bm{\psi}(\theta).
\end{equation}
Substituting \eqref{eqn:polar} into \eqref{eqn:ledlame1}, we can obtain:
\begin{subequations}\label{eqn:cycs}
\begin{numcases}{}
    \rho_i(\theta) = \big(|\cos\theta|^{\gamma_i}/a_i^{\gamma_i}+|\sin\theta|^{\gamma_i}/b_i^{\gamma_i}\big)^{-1/\gamma_i},\\
    z=z_0,
\end{numcases}
\end{subequations}
which are the equations of LED curve $\mathcal{E}_i$ in CyCS. 

Note that the Lam\'e curve in \eqref{eqn:ledlame}--\eqref{eqn:cycs} assumes that its major axis is aligned with the $x$-axis of WCS. In practice, however, some Lam\'e curve-shaped LEDs may not satisfy this alignment. Even so, the proposed algorithm can readily accommodate such case. Suppose the angle between the major axis of the $i$-th LED curve $\mathcal{E}_i$ and the $x$-axis of WCS is $\phi_i$. We can then define a local WCS, denoted $\text{WCS}_i$, obtained by rotating the global WCS counterclockwise around the global $z$-axis by $\phi_i$. Similarly, the corresponding $\text{CyCS}_i$ can be obtained in the same manner. In $\text{CyCS}_i$, the curve $\mathcal{E}_i$ assumes the form of \eqref{eqn:polar} and \eqref{eqn:cycs}. {By rotating \eqref{eqn:cycs} back into the global CyCS, the resulting equation of $\mathcal{E}_i$ is given by:
\begin{equation}\label{eqn:cycsi}
    \bm{e}_i=\overline{\bm{x}}_i+\rho_i(\theta-\phi_i)\cdot\bm{\psi}(\theta-\phi_i),
\end{equation}
where $\bm{e}_i$ represents a point on the curve $\mathcal{E}_i$.} It is therefore evident that employing the CyCS as an auxiliary representation significantly simplifies the subsequent steps of the algorithm. 

Based on the above analysis, the required parameters for each LED include the center $(\overline{x}_i,\overline{y}_i,z_0)^\T$, the major and minor axes $a_i$ and $b_i$, the order $\gamma_i$, and the major-axis orientation $\phi_i$. Accordingly, the LED database is constructed as:
\begin{equation}
    \mathcal{D} = \{(\overline{x}_i,\overline{y}_i,a_i,b_i,\gamma_i,\phi_i)\}_{i=1}^{N'}.
\end{equation}
{As long as the LED parameters are properly calibrated, LC-VLP can consistently perform subsequent contour sampling and CPE for LEDs with arbitrary major-axis orientations.}

\subsection{Correspondence-Free \pnp}\label{sec:fpnp}

In this subsection, our target is to obtain an initial estimate for the subsequent nonlinear optimization. Since the LEDs are Lam\'e curve-shaped rather than simply circular or rectangular, \pnp\ methods\cite{epnp,opnp} become the primary candidates for generating such an initialization. However, as discussed earlier, \pnp\ methods require $\geq4$ pre-calibrated RPs in the environment, and the 3D--2D correspondences must be known\cite{Bai2020}. Although in typical VLP scenarios, the LED centers can be pre-calibrated\cite{vpca,vovlp}, relying solely on LED centers generally requires $\geq4$ LEDs to be simultaneously captured by the camera, which imposes a stringent constraint on the LED layout. To address this limitation, we propose the \freepnp\ method that approximates the solution without relying on enough explicit 3D--2D correspondences, but providing a reliable initialization for the subsequent stage.

First, if the LED centers are pre-calibrated, these 3D--2D correspondences can be directly used; otherwise, the centers of their projection curves can be used as approximations\cite{vpca}. Then, we propose the following theorem:
\begin{theorem}[Invariability of Collinearity]\phantomsection\label{thm:ioc}
    Under the pinhole projection model \eqref{eqn:c2p} and \eqref{eqn:w2c}, given three distinct points on a plane that does not contain the camera optical center, then these points are collinear in WCS \textit{if and only if} their projected points on the image plane are collinear in PCS.
\end{theorem}
\begin{IEEEproof}
    Please refer to Appendix~\ref{app:ioc}.
\end{IEEEproof}

\begin{remark}\phantomsection\label{rem}
For clarity, we illustrate Theorem~\ref{thm:ioc} in Fig.~\ref{fig:syst}. Let the centers of $\LED_1$ and $\LED_2$ be denoted by $G_1$ and $G_2$, respectively, and their projections on the image plane by $G_1'$ and $G_2'$. Consider the point $P_1$ defined as the intersection between segment $\overline{G_1G_2}$ and the boundary of $\LED_1$. According to Theorem~\ref{thm:ioc}, the projection of this point, $P_1'$, must lie on the segment $\overline{G_1'G_2'}$. Since $P_1'$ also lies on the projection curve of $\LED_1$, it follows that $P_1'$ must be the intersection between the projection curve of $\LED_1$ and $\overline{G_1'G_2'}$. Likewise, the point $P_2'$ must be the intersection between the projection curve of $\LED_2$ and $\overline{G_1'G_2'}$. Consequently, we obtain two additional pairs of 3D--2D correspondences beyond the LED centers, i.e., $(P_1,P_1')$ and $(P_2,P_2')$.
\end{remark}

Then, we make the following assumption:
\begin{assumption}\phantomsection\label{ass:dens}
The contour points extracted from the LED's projection curve are densely and approximately uniformly distributed with respect to (w.r.t.) the arc length.
\end{assumption}

\begin{remark}\phantomsection\label{rem:2}
Assumption~\ref{ass:dens} is reasonable since the projection of each LED can typically occupy a sufficiently large number of pixels. Under this condition, edge detection and contour extraction methods\cite{Gong2018} are capable of producing dense and consecutive pixel samples along the boundary. Furthermore, by appropriately adjusting the camera exposure time, the contrast between the LED and the background can be significantly enhanced\cite{vpca}. This can stabilize the edge detection process and consequently result in a contour point set that is more evenly distributed along the projected curve.
\end{remark}

Assume that camera $C$ captures the projections of $N$ LEDs. Based on Theorem~\ref{thm:ioc} and Assumption~\ref{ass:dens}, we aim to sample point sets $\mathcal{E}_i^{\flat}$ and $\mathcal{I}_i^{\flat}$ from the boundary curve $\mathcal{E}_i$ of $\LED_i$ and its projected contour point set $\mathcal{I}_i$, respectively, for applying a \pnp-based initial estimation. Note that the points in $\mathcal{E}_i^{\flat}$ and $\mathcal{I}_i^{\flat}$ are expected to approximately satisfy the 3D--2D projection correspondences, thereby ensuring the estimation reliability.

We first aim to identify the 3D--2D correspondence $(P_i, P_i')$ described in Remark~\ref{rem} for each $\LED_i$. On the image plane, let the center of the projection of $\LED_i$ be denoted by $\overline{\bm{u}}_i$. According to Remark~\ref{rem}, an additional $\LED_\iota$ is required to assist in determining $(P_i, P_i')$; without loss of generality, we set $\iota=(i\bmod N)+1$. Then, from the contour point set $\mathcal{I}_i$, we select the pixel point $\bm{u}_{i0}^\mathrm{p}$ whose direction (\wrt\ $\overline{\bm{u}}_{i}$) is closest to the vector $\overline{\bm{u}}_{\iota}-\overline{\bm{u}}_{i}$ as an approximation of $P_i'$:
\begin{equation}\label{eqn:p_i'}
    \bm{u}_{i0}^\mathrm{p} = \mathop{\arg\min}_{\bm{u}\in\mathcal{I}_i}\,\langle \bm{u}-\overline{\bm{u}}_{i}, \overline{\bm{u}}_{\iota}-\overline{\bm{u}}_{i} \rangle,
\end{equation}
where $\langle\cdot,\cdot\rangle$ represents the included angle between two vectors. In WCS, let the center of $\LED_i$ be denoted by $(\overline{x}_i,\overline{y}_i,z_0)^\T$. Similarly, the $\LED_\iota$ is employed to assist in determining $P_i$. According to \eqref{eqn:cycsi}, the angle between the vector $(\overline{x}_\iota-\overline{x}_i,\overline{y}_\iota-\overline{y}_i,0)^\T$ and the $x_i$-axis of $\text{WCS}_i$, i.e., the polar angle of $P_i$, is given by:
\begin{equation} \label{eqn:beta_i}
    \beta_i = \mathrm{atan2}(\overline{y}_{\iota}-\overline{y}_{i},\overline{x}_{\iota}-\overline{x}_{i})-\phi_i,
\end{equation}
Consequently, the world coordinates of $P_i$ is obtained by:
\begin{equation}\label{eqn:p_i}
    \bm{x}_{i0}=\overline{\bm{x}}_i+\rho_i(\beta_i)\cdot\bm{\psi}(\beta_i).
\end{equation}

From the above analysis, we obtain a 3D--2D correspondence pair $(P_i,P_i')$ for each $\LED_i$. However, we observe that as illustrated in Fig.~\ref{fig:syst}, when $N=2$, the four points $G_1$, $G_2$, $P_1$, and $P_2$ are collinear, which is still insufficient to determine the extrinsics of camera $C$. Therefore, we further propose a contour sampling strategy to augment the set of approximate 3D--2D correspondences.


Specifically, in WCS, we uniformly sample the curve $\mathcal{E}_i$ \wrt\ the polar angle measured from the LED center, and obtain the 3D sampled set $\mathcal{E}_i^{\flat}$. We take the polar angle $\beta_i$ of $P_i$ in $\text{WCS}_i$ as the starting angle, and uniformly sample within the interval $[0, 2\pi)$ to obtain the following set of angles:
\begin{equation} \label{eqn:3dthe}
    \bm{\Theta}_i=\big\{\beta_i,\beta_i+\tfrac{1}{M}2\pi,\dots,\beta_i+\tfrac{M-1}{M}2\pi\big\}\bmod2\pi.
\end{equation}
Then, in a manner consistent with \eqref{eqn:p_i}, the 3D sampling point set $\mathcal{E}_i^{\flat}$ can be obtained as:
\begin{equation} \label{eqn:3dset}
    \mathcal{E}_i^{\flat}= \big\{ \overline{\bm{x}}_i+\rho_i(\theta)\cdot\bm{\psi}(\theta)\mid\theta\in\bm{\Theta}_i \big\}.
\end{equation}

To obtain the corresponding 2D sampling set, we propose the following proposition:
\begin{proposition}\phantomsection\label{pro:samp}
{With moderate camera tilt angles and sufficient camera-to-LED distance, uniform polar-angle sampling along the LED curve yields image projections that approximately correspond to uniform arc-length sampling on the projected contour.}
\end{proposition}
\begin{IEEEproof}
    Please refer to Appendix~\ref{app:samp}.
\end{IEEEproof}

\begin{remark}\phantomsection\label{rem:3}
{In indoor VLP scenarios, the moderate camera tilt angles and sufficient camera-to-LED distances in Proposition~\ref{pro:samp} are generally satisfied, since excessively large tilt angles or extremely short distances usually lead to insufficient visible LEDs within the camera field of view (FoV), which should be avoided in practice.}
\end{remark}

According to Assumption~\ref{ass:dens}, the extracted contour points are approximately uniformly distributed \wrt\ the arc length. Consequently, uniformly sampling these contour points is equivalent to approximately performing uniform arc-length sampling along the projected curve. Moreover, by Proposition~\ref{pro:samp}, once a 3D--2D correspondence pair is determined as the starting points, uniform arc-length sampling on the projected curve corresponds to the projection of uniform polar-angle sampling on the original LED curve. Therefore, within $\mathcal{I}_i$, we assign the index $0$ to $\bm{u}_{i0}^\mathrm{p}$ and re-index the remaining points in $\mathcal{I}_i$ in a counterclockwise order. This $\bm{u}_{i0}^\mathrm{p}$ is recognized as the stating point. Then, we partition the points in $\mathcal{I}_i$ evenly into $M$ segments and select the following indices to form the sampling index set:
\begin{equation} \label{eqn:2dind}
    \mathcal{J}_i=\big\{ 0,\tfrac{1}{M}|\mathcal{I}_i|,\dots,\tfrac{M-1}{M}|\mathcal{I}_i| \big\},
\end{equation}
and thus, the 2D sampling set is defined as $\mathcal{I}^{\flat}_i=\{\bm{u}_{ij}^\mathrm{p}\}_{j\in\mathcal{J}}$, which can be approximately regarded as the 2D projections of the 3D sampling point set $\mathcal{E}^{\flat}_i$.

Finally, for all $N$ captured LEDs, we aggregate all the 3D sampling sets $\mathcal{E}_i^{\flat}$ and the corresponding 2D sampling sets $\mathcal{I}_i^{\flat}$, and apply a \pnp\ algorithm, such as the O\pnp\cite{opnp}, to obtain an initial estimate of the camera extrinsics:
\begin{equation}\label{eqn:pnp}
    \bm{R}^{(0)},\bm{t}^{(0)} = \text{P\textit{MN}P}\big(\textstyle\{\mathcal{E}_i^{\flat}\}_{i=1}^N,\{\mathcal{I}_i^{\flat}\}_{i=1}^N;\bm{K}\big).
\end{equation}

\begin{algorithm}[!t]
    \caption{\freepnp\ Algorithm}
    \small
    \label{alg:freepnp}
    \KwIn{ For $1\leq i\leq N$:
        LED curves $\mathcal{E}_i$; LED centers $\overline{\bm{x}}_i$\;
        \hspace{3em} Projected contour sets $\mathcal{I}_i$; Projected LED centers $\overline{\bm{u}}_i$\;
        \hspace{3em} Camera intrinsic matrix $\bm{K}$\;
    }
    \KwOut{
        Estimate of camera extrinsics: $\bm{R}^{(0)}$ and $\bm{t}^{(0)}$.
    }
    \For{$i=1:N$}{
        Calculate $\iota=(i\bmod N)+1$\;
        Obtain an approximation of $P_i'$, $\bm{u}_{i0}^\mathrm{p}$, by \seqref{eqn:p_i'}\;
        Calculate the starting polar angle $\beta_i$ by \seqref{eqn:beta_i}\;
        Sample points by \seqref{eqn:3dthe} and \seqref{eqn:3dset} to obtain the 3D set $\mathcal{E}_i^{\flat}$\;
        Reorder $\mathcal{I}_i$ counterclockwise with starting point $\bm{u}_{i0}^\mathrm{p}$\;
        Sample points by \seqref{eqn:2dind} to obtain the 2D set $\mathcal{I}_i^{\flat}$\;
    }
    Apply a \pnp\ algorithm to $\{\mathcal{E}_i^{\flat}\}_{i=1}^N$ and $\{\mathcal{I}_i^{\flat}\}_{i=1}^N$ by \seqref{eqn:pnp}\;
\end{algorithm}

The \freepnp\ algorithm is summarized in Algorithm~\ref{alg:freepnp}. Note that in Free\pnp, Assumption~\ref{ass:dens} and Proposition~\ref{pro:samp} introduce an approximation to the 3D--2D correspondences, which inevitably leads to estimation errors. However, these errors are typically within an acceptable range. Moreover, \freepnp\ is only employed to provide an initial estimate of the camera extrinsics for the subsequent nonlinear optimization stage. Therefore, the overall accuracy of LC-VLP is not adversely affected. The nonlinear optimization-based refinement step will be detailed in the next subsection.

\subsection{Nonlinear Optimization-Based Refinement}\label{sec:nobr}

To refine the $\bm{R}^{(0)}$ and $\bm{t}^{(0)}$ obtained in Section~\ref{sec:fpnp}, a commonly adopted approach is to formulate a reprojection-based optimization model\cite{pan2025VLP}. Conventional methods typically require $\geq4$ accurately known pairs of 3D--2D correspondences, and involve projecting 3D points onto the image plane. However, in the considered scenario it is generally not guaranteed that a sufficient number of exact 3D--2D correspondences are available. Moreover, directly projecting Lam\'e curves onto the image plane can result in complicated shape distortions, which makes it difficult to handle the optimization process. 

To address these issues, we first propose a \textit{back-projection strategy}, in which the 2D points on the image plane are back-projected onto the ceiling plane where the optimization is then carried out. On the image plane, for an LED projected point $k'$ with pixel coordinate $\bm{u}_{k'}^{\mathrm{p}}$, its back-projected point $k$ in CCS can be obtained from \eqref{eqn:c2p} as:
\begin{equation}
    \bm{x}_k^\mathrm{c} = z_k^\mathrm{c}\cdot(\bm{K}^{-1}\widetilde{\bm{u}}_{k'}^\mathrm{p}),
\end{equation}
where $z_k^\mathrm{c}$ denotes the depth of point $k$ in CCS and is unknown. Let $\bm{d}_k \triangleq \bm{K}^{-1}\widetilde{\bm{u}}_{k'}^\mathrm{p}$. Then, according to \eqref{eqn:w2c}, the world coordinate of back-projected point $k$ can be expressed as:
\begin{equation}
\label{eqn:k_c2w}
    \widehat{\bm{x}}_k=\bm{R}^\T(\bm{x}_k^\mathrm{c}-\bm{t})=\bm{R}^\T(z_k^\mathrm{c}\bm{d}_k-\bm{t}).
\end{equation}
We next solve for the depth $z_k^\mathrm{c}$. Since the LEDs are mounted on the ceiling, the back-projected point $k$ should lies on the ceiling plane, and thus its vertical coordinate in WCS satisfies $z_k = z_0$. Moreover, let $\bm{R}\triangleq[\bm{r}_1,\bm{r}_2,\bm{r}_3]$ which can be substituted into \eqref{eqn:k_c2w} and yields:
\begin{equation}
    \bm{r}_3^\T\cdot(z_k^\mathrm{c}\bm{d}_k-\bm{t})=z_0,
\end{equation}
where the $z_k^\mathrm{c}$ is solved as:
\begin{equation}
\label{eqn:k_zc}
    z_k^\mathrm{c}=\frac{z_0+\bm{r}_3^\T\cdot\bm{t}}{\bm{r}_3^\T\cdot\bm{d}_k}.
\end{equation}
Then, the back-projected point $\widehat{\bm{x}}_k$ is obtained by substituting \eqref{eqn:k_zc} into \eqref{eqn:k_c2w}:
\begin{equation}\label{eqn:bp}
    \widehat{\bm{x}}_k=\bm{R}^\T\left(\frac{z_0+\bm{r}_3^\T\cdot\bm{t}}{\bm{r}_3^\T\cdot\bm{d}_k}\bm{d}_k-\bm{t}\right).
\end{equation}

After obtaining the back-projected points, we further formulate an NLLS optimization model. Specifically, we define the \textit{algebraic distance}\cite{Guennebaud2007} from point $\bm{x}$ to Lam\'e curve $\mathcal{E}$ as:
\begin{equation}
    \varphi(\bm{x},\mathcal{E})=\big|(x-\overline{x})/{a}\big|^{\gamma} + \big|(y-\overline{y})/{b}\big|^{\gamma} -1.
\end{equation}
By definition, if $\bm{x}$ lies on the curve $\mathcal{E}$, the algebraic distance $\varphi(\bm{x},\mathcal{E})=0$; if $\bm{x}$ lies inside $\mathcal{E}$, $\varphi(\bm{x},\mathcal{E})<0$; otherwise, $\varphi(\bm{x},\mathcal{E})>0$. Moreover, $\varphi(\bm{x},\mathcal{E})$ can provide a coarse measure of the proximity between $\bm{x}$ and $\mathcal{E}$. 

In the absence of noise, the true camera extrinsics would cause all back-projected points $\widehat{\bm{x}}_k$ to lie exactly on their corresponding LED curves $\mathcal{E}_i$, i.e., $\varphi(\widehat{\bm{x}}_k,\mathcal{E}_i)=0,\;\forall\, k'\in\mathcal{I}_i$. In practice, due to image noise and other interferences\cite{vpca,pan2025VLP,vovlp}, the optimal extrinsics $\bm{R}^*$ and $\bm{t}^*$ can be obtained by minimizing the sum of squared algebraic distances from all back-projected points to their corresponding LED curves according to the LS principle. Therefore, we formulate the following NLLS optimization problem:
\begin{subequations}
\begin{gather}
    \mathop{\arg\min}_{\bm{R},\bm{t}} \sum_{i=1}^N 
    \sum_{k'\in\mathcal{I}_i} \varphi^2\big(\widehat{\bm{x}}_k,\mathcal{E}_i\big) \label{eqn:ob}\tag{\theequation}\\
    \mathrm{s.t.}~~\bm{R}\in\mathfrak{SO}(3),\;\bm{t}\in\mathbb{R}^{3},-\bm{R}^\T\bm{t}\in\mathcal{U}, \label{eqn:ob_2}\\
    \frac{1}{|\mathcal{M}|}\sum_{m'\in\mathcal{M}}\big\|\widehat{\bm{x}}_m-\bm{x}_m\big\|^2 \leq \epsilon^2. \label{eqn:ob_3}
\end{gather}
\end{subequations}
Compared to \eqref{eqn:prob1}, in \eqref{eqn:ob_3} we additionally consider the presence of pre-calibrated RPs on the ceiling. {Let $\mathcal{M}$ denote the set of their projected points on the image plane. These projected points $m'$ are also back-projected onto the ceiling plane, and the mean squared back-projection error to their corresponding WCS points $\bm{x}_m$ is constrained to be smaller than $\epsilon^2$. Incorporating this constraint can help eliminate potential ambiguous solutions\cite{vpca} and improve optimization stability during the refinement process.}

In Problem~\eqref{eqn:ob}, rotation matrix $\bm{R}\in\mathfrak{SO}(3)$ contains 9 parameters but has only 3 DoFs. To simplify the optimization, we parameterize $\bm{R}$ using the Rodrigues vector\cite{Valdenebro2016}. Let the Rodrigues vector of $\bm{R}$ be denoted by $\bm{\omega}=\omega\bm{n}$, where $\bm{n}$ is a unit vector. Then, by the \textit{Rodrigues' Formula}\cite{Valdenebro2016}, the rotation matrix can be expressed as:
\begin{equation}
    \bm{R}=\bm{I}+\lfloor\bm{\bm{n}}\rfloor_{\times}\sin\omega+(\lfloor\bm{\bm{n}}\rfloor_{\times})^2(1-\cos\omega),
\end{equation}
where $\lfloor\bm{\bm{n}}\rfloor_{\times}=\left[\begin{smallmatrix}0&-n_z&n_y\\n_z&0&-n_x\\-n_y&n_x&0\end{smallmatrix}\right]$ denotes the skew-symmetric matrix associated with $\bm{n}$. Substituting $\bm{\omega}$ into \eqref{eqn:ob} yields:
\begin{subequations}
\begin{gather}
    \mathop{\arg\min}_{\bm{\omega},\bm{t}} \sum_{i=1}^N 
    \sum_{k'\in\mathcal{I}_i} \varphi^2\big(\widehat{\bm{x}}_k,\mathcal{E}_i\big) \tag{\theequation} \label{eqn:ob2}\\
    \mathrm{s.t.}~~\bm{\omega},\bm{t}\in\mathbb{R}^{3},\;
    -\bm{R}^\T\bm{t}\in\mathcal{U}, \label{eqn:ob2-c1}\\
    \sum_{m'\in\mathcal{M}}\big\|\widehat{\bm{x}}_m-\bm{x}_m\big\|^2 \leq |\mathcal{M}|\cdot\epsilon^2, \label{eqn:ob2-c2}
\end{gather}
\end{subequations}
Then, Problem~\eqref{eqn:ob2} can be solved by using nonlinear optimization algorithms, such as the sequential quadratic programming (SQP) methods\cite{Gill2005}. In this way, the optimal camera extrinsics, $\bm{R}^*$ and $\bm{t}^*$, have been obtained simultaneously. {Note that since Problem~\eqref{eqn:ob2} is inherently non-convex due to the nonlinear perspective projection model and the coupling between $\bm{\omega}$ and $\bm{t}$, the SQP optimization mainly serves as a local refinement stage initialized by Free\pnp.}

To reduce the computational complexity, $\forall\,1\leq i\leq N$, the projected contour point set $\mathcal{I}_i$ can be approximately uniformly subsampled w.r.t. arc length to obtain $\mathcal{I}_i^{(\varsigma)}$, and in this case, only the points in $\mathcal{I}_i^{(\varsigma)}$ are used for back-projection and NLLS optimization. Let the sampling ratio be $0<\varsigma\leq1$, meaning that one point is sampled from $\mathcal{I}_i$ every $\varsigma^{-1}$ points. With this subsampling strategy, the LC-VLP algorithm only requires replacing $\mathcal{I}_i$ in \eqref{eqn:ob} and \eqref{eqn:ob2} with $\mathcal{I}_i^{(\varsigma)}$. Note that as $\varsigma\rightarrow 0$, the number of optimization points decreases, which may lead to reduced optimization improvement, while the computational latency can be correspondingly reduced.

\subsection{Complexity Analysis}

We evaluate the computational overhead of the proposed LC-VLP algorithm by analyzing its theoretical time complexity. In the \freepnp\ stage, for each captured LED, both \eqref{eqn:p_i'} and the reordering of $\mathcal{I}_i$ require traversing all points in $\mathcal{I}_i$, resulting in a complexity of $\mathcal{O}(|\mathcal{I}_i|)$ for each operation. \eqref{eqn:2dind} and \eqref{eqn:3dset} sample $M$ points from $\mathcal{I}_i$ and $\mathcal{E}_i$ to obtain $\mathcal{I}^{\flat}_i$ and $\mathcal{E}^{\flat}_i$, respectively, each with a complexity of $\mathcal{O}(M)$. The above procedures are applied to all $N$ captured LEDs, leading to an complexity of $\mathcal{O}(2\sum_{i=1}^N|\mathcal{I}_i|+2MN)$ for constructing all approximate 3D--2D correspondences. Finally, \eqref{eqn:pnp} applies a \pnp\ solver to the $MN$ 3D--2D correspondences, with a complexity of $\mathcal{O}(MN)$. Overall, the time complexity of \freepnp\ is $\mathcal{O}(2\sum_{i=1}^N|\mathcal{I}_i|+2MN)+\mathcal{O}(MN)=\mathcal{O}(\sum_{i=1}^N|\mathcal{I}_i|+MN)$.

In the refinement stage, the SQP method is employed to solve Problem~\eqref{eqn:ob2}. In each iteration, a total of $\sum_{i=1}^N|\mathcal{I}_i^{(\varsigma)}|=\varsigma\sum_{i=1}^N|\mathcal{I}_i|$ points are back-projected according to \eqref{eqn:bp}, resulting in a per-iteration complexity of $\mathcal{O}(\varsigma\sum_{i=1}^N|\mathcal{I}_i|)$. Moreover, since the dimension of the optimization objective $[\bm{\omega}^\T,\bm{t}^\T]$ and the total number of constraints in \eqref{eqn:ob2-c1} and \eqref{eqn:ob2-c2} are both small constants, the per-iteration complexity of SQP itself can be regarded as a constant\cite{Aoyama2025}. Therefore, the overall time complexity of solving \eqref{eqn:ob2} is $\mathcal{O}(\varsigma I\sum_{i=1}^N|\mathcal{I}_i|)$, where $I$ denotes the number of iterations.

In conclusion, the time complexity of the proposed LC-VLP is expressed as $\mathcal{O}(\sum_{i=1}^N|\mathcal{I}_i|+MN)+\mathcal{O}(\varsigma I\sum_{i=1}^N|\mathcal{I}_i|)=\mathcal{O}(\varsigma I\sum_{i=1}^N|\mathcal{I}_i|)$, where $M<|\mathcal{I}_i|\ll \varsigma I|\mathcal{I}_i|$.

\section{Simulations and Experiments}\label{sec:simexp}

In this section, we verify the performance and feasibility of our proposed LC-VLP algorithm via simulations and experiments. To quantitatively evaluate the positioning performance of LC-VLP and baselines, we respectively define the position error $\varepsilon_p$ and rotation error $\varepsilon_r$ as:
\begin{equation}
    \varepsilon_{p} = \big\|\widehat{\bm{x}}-\bm{x}\big\|,
\end{equation}
\begin{equation}
    \varepsilon_{r} = \tfrac{180}{\pi}\arccos \tfrac{1}{2}\big({\mathrm{tr}(\widehat{\bm{R}}\bm{R}^\T)-1}\big),
\end{equation}
where $\bm{x}$ and $\widehat{\bm{x}}$ denote the true and estimated positions, respectively, and $\bm{R}$ and $\widehat{\bm{R}}$ denote the corresponding rotation matrices. Based on these, we define the following metrics: (i) Mean of position/rotation errors (MPE/MRE); (ii) 50\textsuperscript{th}/90\textsuperscript{th} percentile of position/rotation errors (P50/P90/R50/R90); and (iii) standard deviation (STD) of position errors.

The following algorithms are chosen as baselines: 
\begin{itemize}
  \item \textbf{V-PCA\cite{vpca}}, a SoTA algorithm based on circular LEDs. This method estimates the camera extrinsics and position using a purely geometric method, and requires at least two circular LEDs to be captured simultaneously.
  \item \textbf{VLC-\pnp\cite{vlcpnp}}, a SoTA algorithm based on rectangular LEDs. This method can estimate the camera pose using only a single LED; however, it relies on an additional magnetometer to assist in establishing the correct 3D--2D correspondences. In the simulations, we assume that there is no electromagnetic interference and the magnetometer measurements are perfectly accurate.
  \item \textbf{O\pnp\cite{opnp}}, an advanced high-accuracy \pnp\ solver. In the simulations, to enable the implementation of this method, four extra RPs are placed at the endpoints of the major and minor axes of each LED, respectively. This additional setup is not used by other baselines.
  \item {\textbf{ISA-SVG\cite{isasvg}}, an advanced IMU-assisted algorithm. This method requires two LEDs for 3D positioning and treats the LEDs as point sources, making it independent of the LED shape. However, it requires prior pose information provided by an IMU and is incapable of estimating the camera orientation. Therefore, the rotation error of ISA-SVG is not included in the results.}
\end{itemize}

{We note that to ensure a fair comparison, in all the following simulations and experiments we consistently guarantee that at least two LEDs are fully captured by the camera.}

\subsection{Simulation Setup and Results}

\begin{table}[!t]
\small
\centering
\caption{{Simulation Parameters}}
\label{tab:sim_para}
\setlength{\extrarowheight}{2pt}
\setlength{\tabcolsep}{5pt}{
\scalebox{0.9}{
\begin{tabular}{|cc|P{4.5cm}|} \hline \rowcolor[HTML]{EFEFEF} 
\multicolumn{2}{|c|}{\textbf{Parameter}} & \textbf{Value} \\\hline
\multicolumn{1}{|c|}{\multirow{5}{*}{Platform}} & Room size & $6\m\times8\m\times3\m$ \\\cline{2-3} 
\multicolumn{1}{|c|}{} & LED semi-angle & $\Phi_{1/2}=60^\circ$ \\\cline{2-3} 
\multicolumn{1}{|c|}{} & LED positions & $(2,2,3)\m$, $(2,6,3)\m$, $(4,2,3)\m$, $(4,6,3)\m$ \\\cline{2-3} 
\multicolumn{1}{|c|}{} & LED power & $P_t=3\,\mathrm{W}$ \\\hline
\multicolumn{1}{|c|}{\multirow{3}{*}{Camera}} & Focal lengths & $f_x=f_y=800\,\mathrm{px}$ \\\cline{2-3} 
\multicolumn{1}{|c|}{} & Principle point & $(u_0,u_0)=(320,240)\,\mathrm{px}$ \\\cline{2-3} 
\multicolumn{1}{|c|}{} & Pixel physical size & $\delta x/\delta u=\delta y/\delta v=1.25\,\mathrm{\mu m/px}$ \\\hline
\multicolumn{1}{|c|}{LC-VLP} & Sampling ratio & $\varsigma=1$ \\\hline
\end{tabular}
}}
\end{table}

In the simulations, we consider a cuboid room in which four LEDs are installed at fixed positions on the ceiling. A camera with randomly generated position and orientation is used to capture the LEDs. To ensure the proper execution of positioning algorithms, the generated camera pose is always constrained such that at least two LEDs are visible. Unless otherwise specified, key system parameters are summarized in Table~\ref{tab:sim_para}. {The camera image noise is modeled as a zero-mean white Gaussian noise\cite{vpca,bai2021a,vp4l,vovlp,isasvg,He2024VLP,pan2025VLP,epnp,opnp} with an STD of $2\,\mathrm{px}$, which is widely adopted in both fields of VLP and CV. For the IMU used in the ISA-SVG algorithm, the pitch and roll noises are modeled as uniformly distributed random variables following $\mathcal{U}[-1.5^\circ,1.5^\circ]$, while the yaw noise is modeled as $\mathcal{U}[-15^\circ,15^\circ]$\cite{Hao2019,vpca,vovlp}.} All statistical results in the simulations are averaged over $10^4$ independent iterations to ensure universality. {Key parameter settings of the SQP method are: max\_iterations = 300, step\_tolerance = 1e-8, optimality\_tolerance = 1e-8, and max\_function\_evaluations = 5e3.}

\begin{table*}[!t]
\small
\centering
\caption{{Simulation Statistical Results}}
\label{tab:sim_res}
\setlength{\extrarowheight}{2pt}
\setlength{\tabcolsep}{5pt}{
\scalebox{0.9}{
\begin{tabular}{|P{2.7cm}|P{2.4cm}|c|c|c|c|c|c|c|}
\hline \rowcolor[HTML]{EFEFEF}
\textbf{Scenario} & \textbf{Algorithm} & \textbf{MPE [cm] $\downarrow$} & \textbf{P50 [cm] $\downarrow$} & \textbf{P90 [cm] $\downarrow$} & \textbf{STD [cm] $\downarrow$} & \textbf{MRE [$^\circ$] $\downarrow$} & \textbf{R50 [$^\circ$] $\downarrow$} & \textbf{R90 [$^\circ$] $\downarrow$} \\\hline

& \textbf{LC-VLP (Ours)} & \textbf{2.25} & \textbf{1.76} & \textbf{4.70} & \textbf{1.91} & \textbf{0.28} & \textbf{0.23} & \textbf{0.59} \\\cline{2-9}
& V-PCA\cite{vpca} & 3.93 & 3.39 & 7.38 & 2.54 & 0.42 & 0.38 & 0.75 \\\cline{2-9}
& O\pnp\cite{opnp} & 5.21 & 3.50 & 10.82 & 8.79 & 0.65 & 0.50 & 1.36 \\\cline{2-9}
\multirow{-4}{*}{\shortstack{Scenario~\ref{sc:1}\\(Circular LEDs)}} & ISA-SVG\cite{isasvg} & 19.25 & 16.48 & 34.65 & 12.45 & N/A & N/A & N/A \\\hline

& \textbf{LC-VLP (Ours)} & \textbf{2.71} & \textbf{2.05} & \textbf{5.20} & \textbf{2.80} & \textbf{0.33} & \textbf{0.26} & \textbf{0.66} \\\cline{2-9}
& VLC-\pnp\cite{vlcpnp} & 5.34 & 3.85 & 11.78 & 4.82 & 1.48 & 1.15 & 3.04 \\\cline{2-9}
& O\pnp\cite{opnp} & 6.08 & 4.07 & 14.51 & 9.43 & 0.75 & 0.58 & 1.62 \\\cline{2-9}
\multirow{-4}{*}{\shortstack{Scenario~\ref{sc:2}\\(Rectangular LEDs)}} & ISA-SVG\cite{isasvg} & 19.16 & 16.46 & 34.21 & 12.77 & N/A & N/A & N/A \\\hline
\end{tabular}
}}
\end{table*}

\begin{figure}[!t]
\begin{subfigure}{0.235\textwidth}\centering
    \includegraphics[width=\textwidth]{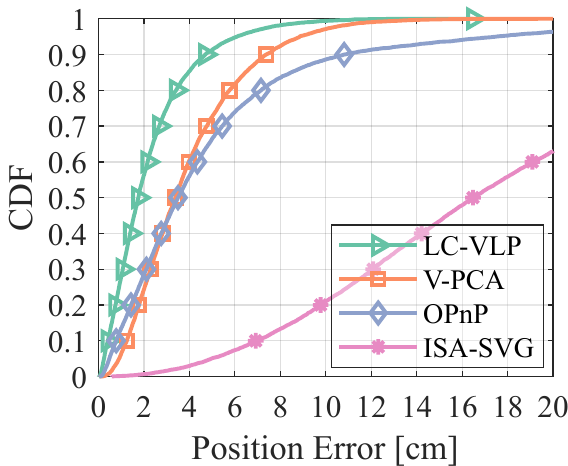}
    \caption{}\label{fig:sim_c_pos}
\end{subfigure}
\begin{subfigure}{0.235\textwidth}\centering
    \includegraphics[width=\textwidth]{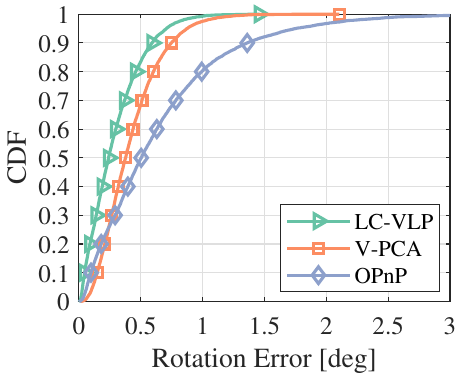}
    \caption{}\label{fig:sim_c_rot}
\end{subfigure}
\begin{subfigure}{0.235\textwidth}\centering
    \includegraphics[width=\textwidth]{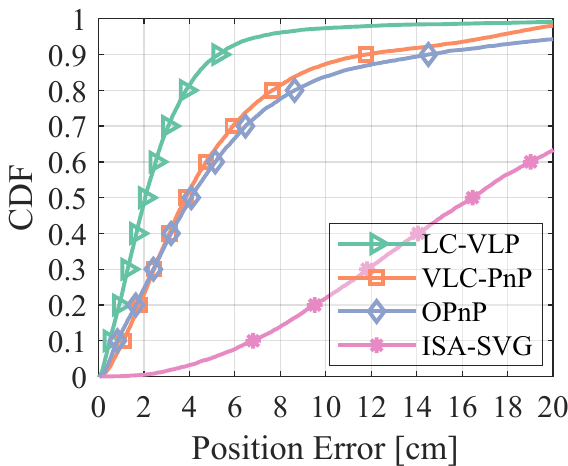}
    \caption{}\label{fig:sim_r_pos}
\end{subfigure}
\begin{subfigure}{0.235\textwidth}\centering
    \includegraphics[width=\textwidth]{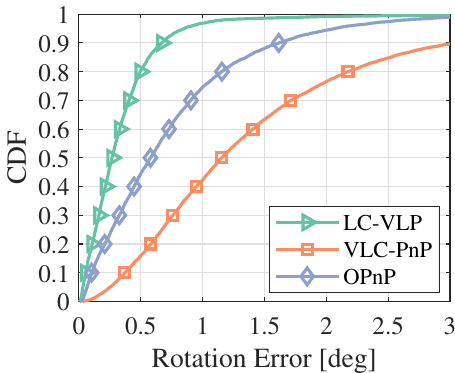}
    \caption{}\label{fig:sim_r_rot}
\end{subfigure}
\centering
\caption{{The simulation CDFs of position and rotation errors. (a) and (b) are the results of Scenario~\ref{sc:1}, while (c) and (d) are the results of Scenario~\ref{sc:2}.}}
\label{fig:sim_cdf}
\end{figure}

\subsubsection{Homogeneous LED-Shape Scenarios} \label{sec:hLED}
We first consider scenarios where all LEDs share the same shape, with circular and rectangular LEDs taken as representative examples:
\begin{enumerate}[label=\textbullet,ref=\Alph*]
    \item \textit{Scenario~\Alph{enumi}:} Circular LEDs with a radius of $0.15\cm$. \label{sc:1}
    \item \textit{Scenario~\Alph{enumi}:} Rectangular LEDs with a semi-major axis of $0.15\cm$ and a semi-minor axis of $0.12\cm$. \label{sc:2}
\end{enumerate}
{Note that, in these two scenarios, we primarily focus on comparing LC-VLP with V-PCA and VLC-\pnp\ under their corresponding LED-shape settings. This is because, even if V-PCA and VLC-\pnp\ were applied to other LED-shape scenarios, their performance would degrade significantly due to the mismatch between the assumed and actual LED geometries. This will be further analyzed in Scenario~\ref{sc:4}.}

Fig.~\ref{fig:sim_cdf} shows the cumulative distribution functions (CDFs) of the position and rotation errors for different algorithms. It can be observed that, in both scenarios, LC-VLP consistently outperforms all baseline methods. Specifically, in Scenario~\ref{sc:1}, LC-VLP achieves a 95\textsuperscript{th}-percentile position accuracy of $6\cm$ and an 83\textsuperscript{rd}-percentile rotation accuracy of $0.5^\circ$, whereas V-PCA only attains 82\textsuperscript{nd}- and 70\textsuperscript{th}-percentile accuracies, respectively. This performance gap arises because V-PCA is a purely geometric solution, whose accuracy and stability are inherently lower than those of the iterative optimization-based LC-VLP. {By comparison, O\pnp\ and ISA-SVG exhibit larger errors, since O\pnp\ does not exploit the fixed shape constraints among the RPs, and ISA-SVG relies on IMU which introduces additional pose errors.}

{As shown in Figs.~\ref{fig:sim_cdf}\subref{fig:sim_r_pos} and \subref{fig:sim_r_rot}, in Scenario~\ref{sc:2}, LC-VLP also achieves 93\textsuperscript{rd}- and 80\textsuperscript{th}-percentile position and rotation accuracies of $6\cm$ and $0.5^\circ$, respectively, which significantly outperforms VLC-\pnp, O\pnp\ and ISA-SVG.} Note that although both VLC-\pnp\ and O\pnp\ are \pnp-based methods, their performances differ, since O\pnp\ is optimized for coplanar RPs, whereas VLC-\pnp\ adopts a conventional \pnp\ solver. When all RPs are coplanar, the latter is prone to accuracy degradation caused by depth ambiguity\cite{opnp}. Therefore, although VLC-\pnp\ applies nonlinear refinement, it remains more susceptible to large rotation errors. In contrast, LC-VLP incorporates not only the corner points but also the boundary points along LED edges into calculation. Therefore, compared with \pnp-based methods, this richer geometric constraint significantly improves both its accuracy and robustness.

\begin{figure}[!t]\centering
    \includegraphics[width=0.34\textwidth]{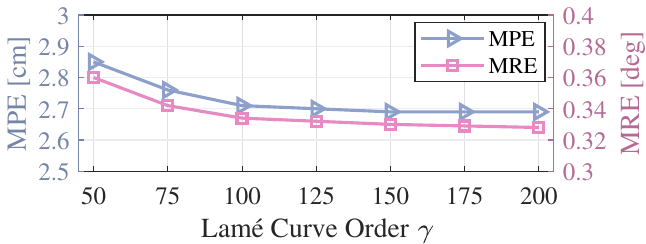}
    \caption{{Comparison of MPE and MRE of LC-VLP under different Lam\'e curve orders $\gamma$ in Scenario~\ref{sc:2}.}}
    \label{fig:order}
\end{figure}

The above statistical results are summarized in Table~\ref{tab:sim_res}. It can be observed that LC-VLP achieves MPEs of $2.25\cm$ and $2.71\cm$ in Scenarios~\ref{sc:1} and \ref{sc:2}, respectively, corresponding to reductions of $43\%$ and $49\%$ compared with V-PCA and VLC-\pnp. Likewise, in both scenarios, the MRE of LC-VLP is reduced by $25\%$ and $78\%$, respectively, relative to V-PCA and VLC-\pnp. In terms of positioning stability, LC-VLP consistently yields the lowest STD, indicating that it maintains a more balanced accuracy across different spatial locations and camera poses, thereby demonstrating superior robustness. Moreover, other metrics reported in Table~\ref{tab:sim_res} further confirm that LC-VLP outperforms other baselines in all aspects, exhibiting high accuracy and stability in both position and orientation estimation.

{In Scenario~\ref{sc:2}, we consistently use $\gamma = 100$ to approximate rectangular LEDs with Lam\'e curves. In Fig.~\ref{fig:order}, we further illustrate the impact of different curve orders $\gamma$ on the performance of LC-VLP. It can be observed that, as $\gamma$ increases, both MPE and MRE of LC-VLP gradually decrease and eventually converge. This is because a larger order makes the Lam\'e curve closer to the actual rectangle, thereby improving the accuracy in both the FreeP\textit{n}P and nonlinear optimization stages. Moreover, when $\gamma > 100$, the Lam\'e curve is already sufficiently close to a rectangle, and thus further increasing $\gamma$ brings only marginal improvement. To avoid floating-point precision loss caused by excessively large exponents, in practical applications, choosing $\gamma\geq100$ is generally sufficient.}

\begin{figure*}[!t]
\begin{subfigure}{0.235\textwidth}\centering
    \includegraphics[width=\textwidth]{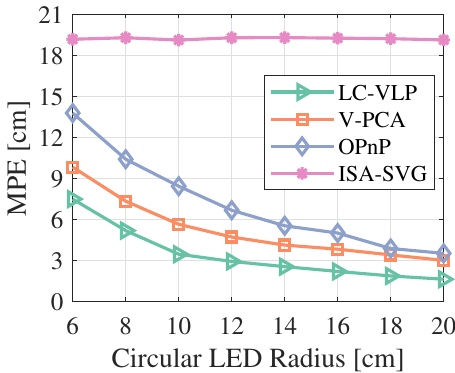}
    \caption{}\label{fig:sim_rr_c}
\end{subfigure}
\begin{subfigure}{0.235\textwidth}\centering
    \includegraphics[width=\textwidth]{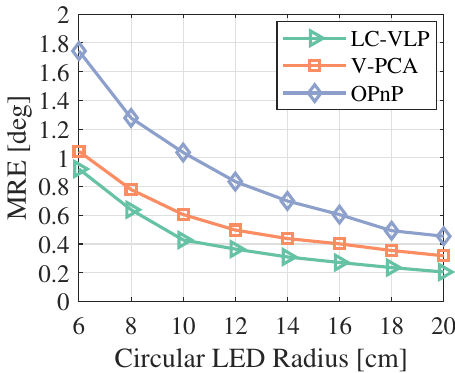}
    \caption{}\label{fig:sim_rr_c_rot}
\end{subfigure}
\begin{subfigure}{0.235\textwidth}\centering
    \includegraphics[width=\textwidth]{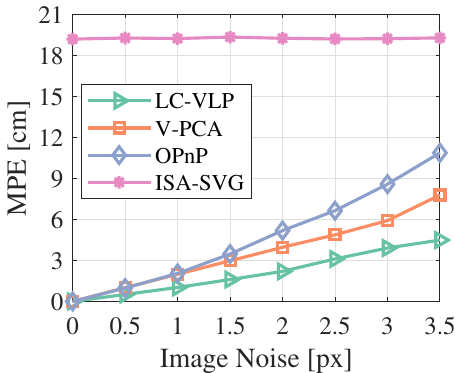}
    \caption{}\label{fig:sim_nn_c}
\end{subfigure}
\begin{subfigure}{0.235\textwidth}\centering
    \includegraphics[width=\textwidth]{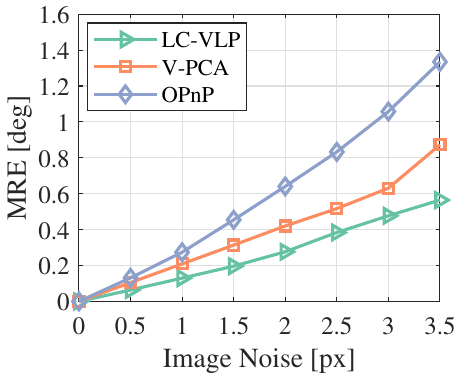}
    \caption{}\label{fig:sim_nn_c_rot}
\end{subfigure}
\begin{subfigure}{0.235\textwidth}\centering
    \includegraphics[width=\textwidth]{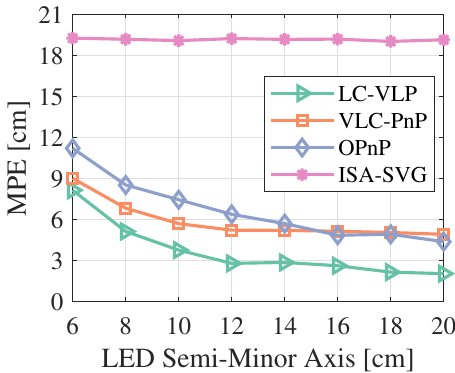}
    \caption{}\label{fig:sim_rr_r}
\end{subfigure}
\begin{subfigure}{0.235\textwidth}\centering
    \includegraphics[width=\textwidth]{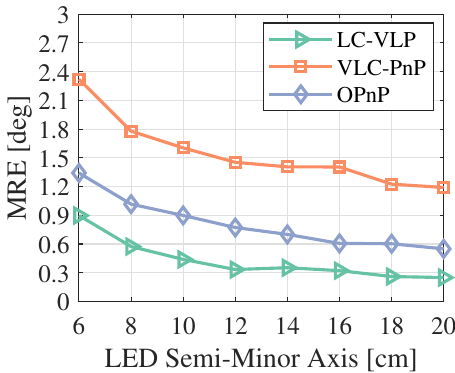}
    \caption{}\label{fig:sim_rr_r_rot}
\end{subfigure}
\begin{subfigure}{0.235\textwidth}\centering
    \includegraphics[width=\textwidth]{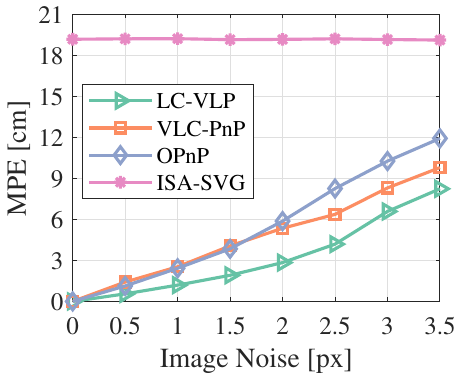}
    \caption{}\label{fig:sim_nn_r}
\end{subfigure}
\begin{subfigure}{0.235\textwidth}\centering
    \includegraphics[width=\textwidth]{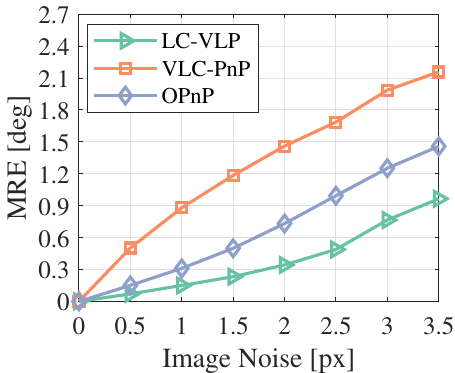}
    \caption{}\label{fig:sim_nn_r_rot}
\end{subfigure}
\centering
\caption{{The simulation results across different LED sizes and image noise STDs. (a)--(d) correspond to Scenario~\ref{sc:1}: (a) and (b) show the MPE and MRE versus the LED radius, while (c) and (d) show the MPE and MRE versus the image noise. (e)--(h) correspond to Scenario~\ref{sc:2}: (e) and (f) show the MPE and MRE versus the LED semi-minor axis, while (g) and (h) show the MPE and MRE versus the image noise.}}
\label{fig:sim}
\end{figure*}

In Fig.~\ref{fig:sim}, we investigate the impact of LED size and image noise on positioning performance. Note that when varying the LED size, the circular LED radius is no longer fixed in Scenario~\ref{sc:1}, while in Scenario~\ref{sc:2} the rectangular LED’s major axis is fixed and the minor axis is treated as the variable. We can observe that, as the LED size increases, both position and rotation estimation errors of all algorithms consistently decrease; conversely, increasing image noise leads to higher errors for all methods. In particular, in Figs.~\ref{fig:sim}\subref{fig:sim_rr_c} and \subref{fig:sim_rr_c_rot}, as the radius of the circular LED varies, the MPE of LC-VLP ranges from $7.48\cm$ to $1.62\cm$. By comparison, the 2\textsuperscript{nd}-best method, V-PCA, achieves an MPE ranging from $9.82\cm$ to $3.00\cm$. In terms of rotation accuracy, LC-VLP attains an MRE between $0.92^\circ$ and $0.20^\circ$, whereas the MRE of V-PCA varies from $1.00^\circ$ to $0.32^\circ$. In Figs.~\ref{fig:sim}\subref{fig:sim_nn_c} and \subref{fig:sim_nn_c_rot}, as the image noise increases, the MPE of LC-VLP gradually rises to $4.48\cm$ and the MRE increases to $0.56^\circ$. In contrast, V-PCA and O\pnp\ are more sensitive to image noise, exhibiting noticeably larger increases in both position and rotation errors.

Figs.~\ref{fig:sim}\subref{fig:sim_rr_r} and \subref{fig:sim_rr_r_rot} illustrate the variations of MPE and MRE with respect to the LED semi-minor axis in Scenario~\ref{sc:2}. As the semi-minor axis increases, the MPE of LC-VLP decreases from $8.13\cm$ to $2.03\cm$, and its MRE decreases from $0.90^\circ$ to $0.25^\circ$, which consistently achieve the best performance across different LED sizes. From Figs.~\ref{fig:sim}\subref{fig:sim_nn_r} and \subref{fig:sim_nn_r_rot}, it can be observed that in Scenario~\ref{sc:2}, as the image noise increases, the MPE of LC-VLP rises from $0$ to $8.23\cm$ and the MRE increases from $0$ to $0.96^\circ$. Nevertheless, LC-VLP exhibits the smallest error growth among all baselines. In contrast, although VLC-\pnp\ achieves the 2\textsuperscript{nd}-best position accuracy, its rotation error is higher than that of O\pnp. Overall, both VLC-\pnp\ and O\pnp\ are consistently outperformed by LC-VLP.

{Moreover, Figs.~\ref{fig:sim}\subref{fig:sim_rr_c}, \subref{fig:sim_nn_c}, \subref{fig:sim_rr_r}, and \subref{fig:sim_nn_r} also show that the performance of ISA-SVG varies only marginally with the LED size and image noise, remaining around $19\cm$ in most cases. This is because ISA-SVG treats LEDs as point sources and ignores their geometric features; additionally, compared with image noise, the orientation noise introduced by the IMU is significantly more dominant and becomes the primary source of position error. In general, LC-VLP also significantly outperforms ISA-SVG.}

\begin{figure}[!t]\centering
    \includegraphics[width=0.27\textwidth]{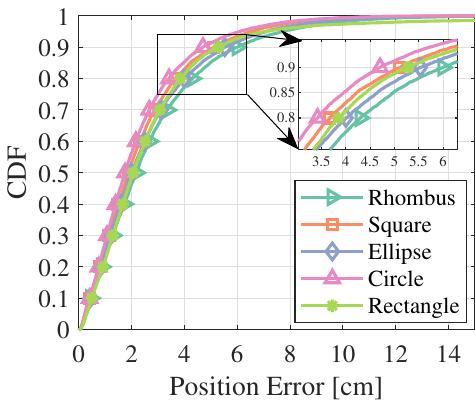}
    \caption{The simulation results of LC-VLP across different LED shapes in Scenario~\ref{sc:3}.}
    \label{fig:cdfs}
\end{figure}

\begin{figure}[!t]
\begin{subfigure}{0.235\textwidth}\centering
    \includegraphics[width=\textwidth]{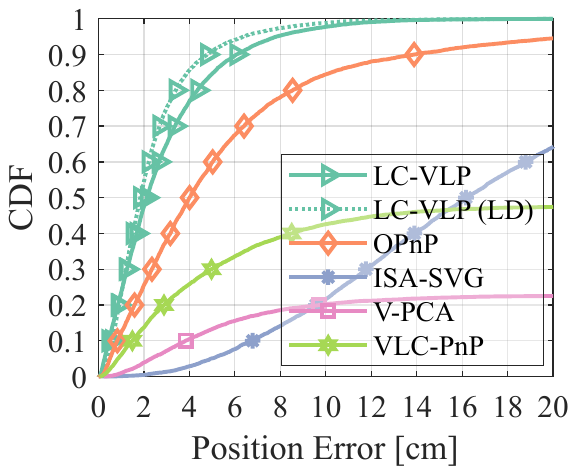}
    \caption{}\label{fig:cdfs_hh_pos}
\end{subfigure}
\begin{subfigure}{0.235\textwidth}\centering
    \includegraphics[width=\textwidth]{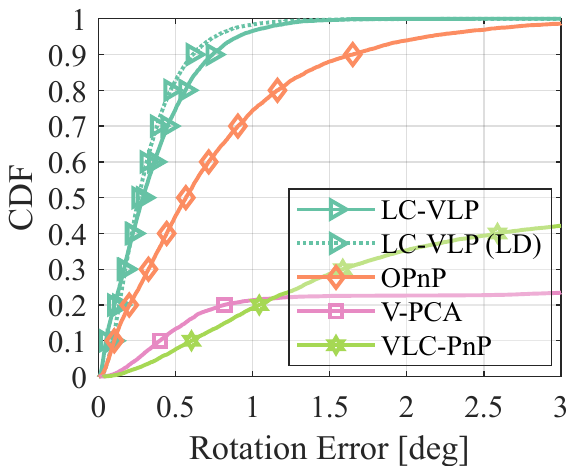}
    \caption{}\label{fig:cdfs_hh_rot}
\end{subfigure}
\centering
\caption{{The simulation results in Scenario~\ref{sc:4}.}}
\label{fig:cdfs_hh}
\end{figure}

We further consider the following scenario where LEDs of other common shapes are deployed:
\begin{enumerate}[label=\textbullet,ref=\Alph*,start=3]
    \item \textit{Scenario~\Alph{enumi}:} Lam\'e curve LEDs with $\gamma=1$ (rhombus or square) or $\gamma=2$ (ellipse). Note that, in the square case, the LED’s major and minor axes are identical to those in Scenario~\ref{sc:1}; whereas in the rhombus or ellipse cases, the major and minor axes are identical to Scenario~\ref{sc:2}. \label{sc:3}
\end{enumerate}

In Fig.~\ref{fig:cdfs} we use CDFs to illustrate the impact of different LED shapes on the performance of LC-VLP under Scenario~\ref{sc:3}. As shown, for rhombic, square, elliptical, circular, and rectangular LEDs, LC-VLP achieves MPEs of $2.80\cm$, $2.41\cm$, $2.60\cm$, $2.23\cm$, and $2.91\cm$, respectively. These results indicate that, across Lam\'e curve-shaped LEDs with different $\gamma$'s, the MPE of LC-VLP consistently remains below $3\cm$, and the performance variation among different LED shapes is relatively small. We note that square and circular LEDs yield higher accuracy than the other shapes, which is attributable to their larger areas in Scenario~\ref{sc:3}. This observation is consistent with the results in Figs.~\ref{fig:sim}\subref{fig:sim_rr_c}, \subref{fig:sim_rr_c_rot}, \subref{fig:sim_rr_r}, and \subref{fig:sim_rr_r_rot}, where larger LED areas are shown to improve the accuracy of LC-VLP.

\subsubsection{Heterogeneous LED-Shape Scenarios}

In Section~\ref{sec:hLED}, we simulated and compared the performance of different algorithms under homogeneous LED-shape scenarios. However, unlike V-PCA and VLC-\pnp, the proposed LC-VLP is capable of achieving generic positioning when multiple LED shapes coexist in the environment. To evaluate this capability, we consider a complex scenario with heterogeneous LED layouts:
\begin{enumerate}[label=\textbullet,ref=\Alph*,start=4]
    \item \textit{Scenario~\Alph{enumi}:} Four LEDs with different shapes (rhombus, ellipse, circle, and rectangle as defined in Scenarios~\ref{sc:1}--\ref{sc:3}) are deployed at the four LED positions listed in Table~\ref{tab:sim_para}. \label{sc:4}
\end{enumerate}

{Figs.~\ref{fig:cdfs_hh}\subref{fig:cdfs_hh_pos} and \subref{fig:cdfs_hh_rot} compare the position and rotation errors of LC-VLP, O\pnp, ISA-SVG, V-PCA and VLC-\pnp\ in Scenario~\ref{sc:4}.} As shown, LC-VLP achieves a 90\textsuperscript{th}-percentile position accuracy of $6\cm$ and a 75\textsuperscript{th}-percentile rotation accuracy of $0.5^\circ$. The MPE and MRE of LC-VLP are $2.85\cm$ and $0.35^\circ$, respectively, which outperforms O\pnp\ by over $53\%$. {We also apply V-PCA and VLC-\pnp\ for positioning. Note that, for V-PCA, we assume the required circular LED parameters are fixed as those of the circular LED in Scenario~\ref{sc:1}; for VLC-\pnp\, since it relies on extracting the vertices of rectangular LEDs, while circular and elliptical LEDs do not possess such vertices, positioning is regarded as failed whenever no rectangular LED is captured (i.e., position error is treated as $\infty$ and rotation error as $180^\circ$). It can be observed that the performance of both V-PCA and VLC-\pnp\ degrades significantly compared with that in Scenarios~\ref{sc:1} and \ref{sc:2}. In particular, the MPE of V-PCA exceeds $3\m$, while VLC-\pnp\ exhibits a failure rate of approximately $50\%$. This is because the scenario contains light sources that are inconsistent with the assumptions of V-PCA and VLC-\pnp, both of which are designed only for a single specific LED geometry. In contrast, LC-VLP still exhibit superior positioning capability in heterogeneous LED-shape scenarios.}

{In Scenario~\ref{sc:4}, we also considered a lower-density (LD) LED deployment. Specifically, the LEDs located at $(2,2,3)\m$ and $(2,6,3)\m$ were removed. The results are shown as LC-VLP (LD) in Fig.~\ref{fig:cdfs_hh}. We can observe that neither the positioning error nor the rotation error increases, which is because the performance of LC-VLP is not directly determined by the total number of LEDs, but depends on the number and geometric shapes of the LEDs captured by the camera.}

\subsection{Experimental Implementation}

\begin{figure}[!t]
\begin{subfigure}{0.37\textwidth}
    \centering
    \includegraphics[width=\textwidth]{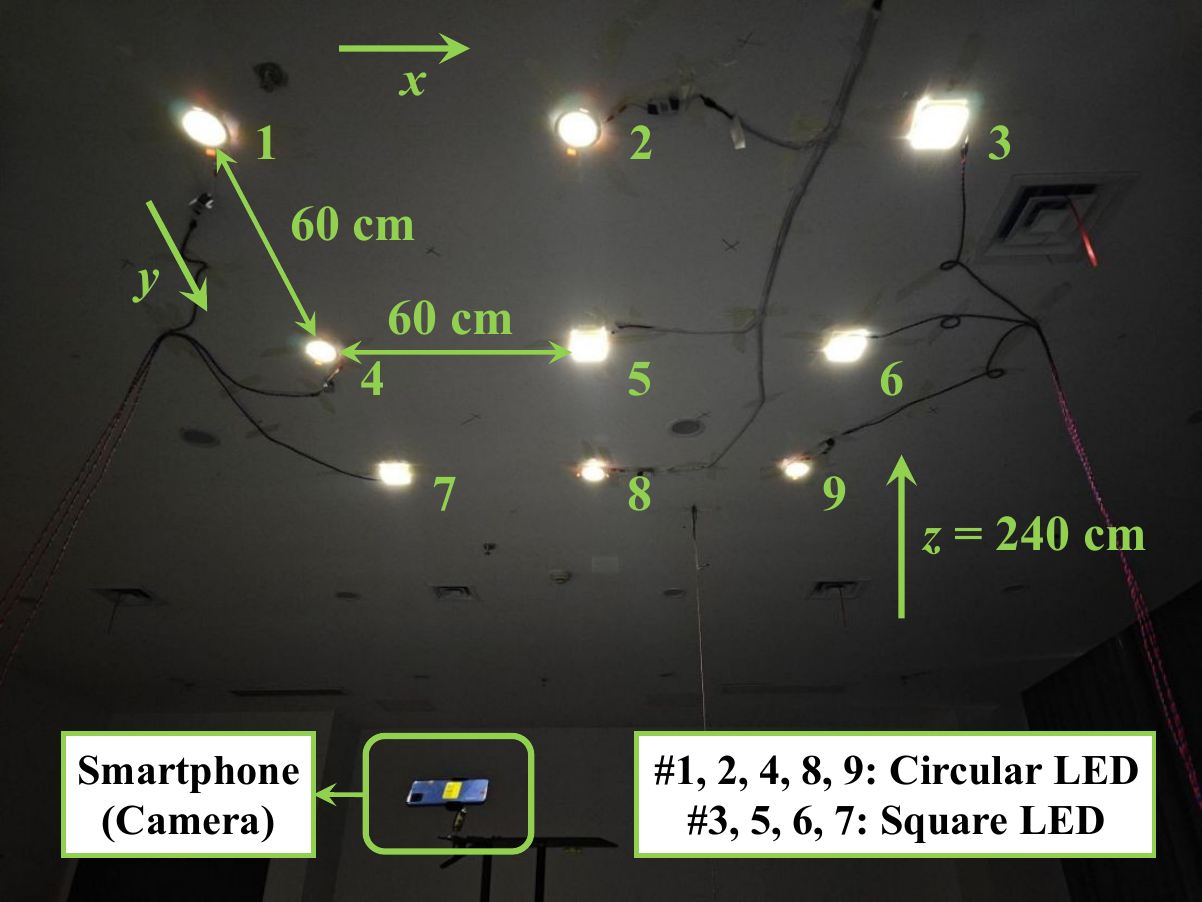}
    \caption{The prototype of LC-VLP positioning system.}\label{fig:proto}
    \vspace{1ex}
\end{subfigure}
\begin{subfigure}{0.47\textwidth}\centering
    \includegraphics[width=\textwidth]{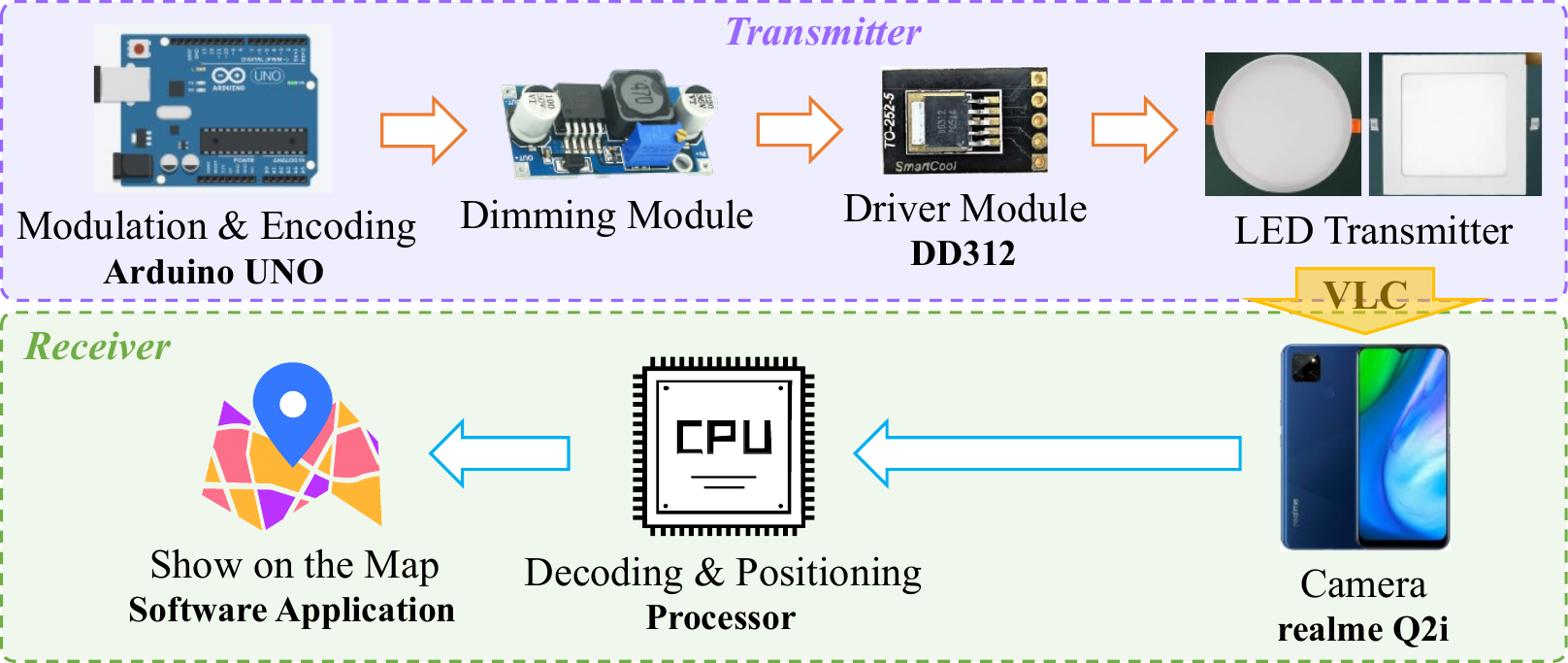}
    \caption{The system flowchart.}\label{fig:imp}
\end{subfigure}
\centering
\caption{{The experimental implementation of LC-VLP. Note that in (a), to highlight the LEDs and the smartphone, this image was captured with an aperture of f/1.6, an exposure time of 0.05\,s, and an ISO of 125 to reduce the surrounding brightness.}}
\label{fig:exp_imp}
\end{figure}

To further verify the feasibility of the proposed LC-VLP, we built an experimental platform, which is a heterogeneous LED-shape scenario. As shown in Fig.~\ref{fig:exp_imp}\subref{fig:proto}, the experiment was conducted in a cuboid space of size $2.4\m\times 2.4\m\times 2.4\m$, where a total of 9 circular (``\circLED") and square (``\rectLED") LEDs are installed on the ceiling. {The key experimental parameters are listed in Table~\ref{tab:exp_para}, with the same SQP settings as those used in the simulations. All experiments were conducted under normal indoor ambient lighting conditions.}

\begin{table}[!t]
\small
\centering
\caption{{Experimental Parameters}}
\label{tab:exp_para}
\setlength{\extrarowheight}{2pt}
\setlength{\tabcolsep}{5pt}{
\scalebox{0.9}{
\begin{tabular}{|ccP{2cm}|P{4.3cm}|}
\hline
\rowcolor[HTML]{EFEFEF} 
\multicolumn{3}{|c|}{\cellcolor[HTML]{EFEFEF} \textbf{Parameter}} & \textbf{Value} \\ \hline
\multicolumn{1}{|c|}{} & \multicolumn{1}{c|}{} & Semi-angle & $\Phi_{1/2}=67.5^\circ$ \\\cline{3-4} 
\multicolumn{1}{|c|}{} & \multicolumn{1}{c|}{} & Power & $P_t=4\,\mathrm{W}$ \\\cline{3-4} 
\multicolumn{1}{|c|}{} & \multicolumn{1}{c|}{} & Radius & $R_L=3.4\cm$ \\\cline{3-4} 
\multicolumn{1}{|c|}{} & \multicolumn{1}{c|}{\multirow{-4}{*}{\circLED}} & Positions \scalebox{0.9}{(Height $240\cm$)} & $(60,60)\cm$, $(120,60)\cm$, $(60,120)\cm$, $(120,180)\cm$, $(180,180)\cm$ \\\cline{2-4} 
\multicolumn{1}{|c|}{} & \multicolumn{1}{c|}{} & Semi-angle & $\Phi_{1/2}=67.5^\circ$ \\\cline{3-4} 
\multicolumn{1}{|c|}{} & \multicolumn{1}{c|}{} & Power & $P_t=7\,\mathrm{W}$ \\\cline{3-4} \multicolumn{1}{|c|}{} & \multicolumn{1}{c|}{} & Edge length & $D_L=8.0\cm$ \\\cline{3-4} 
\multicolumn{1}{|c|}{\multirow{-8}{*}{LED}} & \multicolumn{1}{c|}{\multirow{-4}{*}{\rectLED}} & Positions \scalebox{0.9}{(Height $240\cm$)} & $(180,60)\cm$, $(120,120)\cm$, $(180,120)\cm$, $(60,180)\cm$ \\\hline
\multicolumn{1}{|c|}{} & \multicolumn{2}{c|}{Model} & Rear camera of realme Q2i \\\cline{2-4} 
\multicolumn{1}{|c|}{} & \multicolumn{2}{c|}{} & $f=3.462\,\mathrm{mm}$ \\\cline{4-4} 
\multicolumn{1}{|c|}{} & \multicolumn{2}{c|}{} & $\delta x/\delta u=\delta y/\delta v=1.12\,\mathrm{\mu m/px}$ \\\cline{4-4} 
\multicolumn{1}{|c|}{} & \multicolumn{2}{c|}{\multirow{-3}{*}{Intrinsics}} & $(u_0,u_0)=(2080,1560)\,\mathrm{px}$ \\\cline{2-4} 
\multicolumn{1}{|c|}{\multirow{-5}{*}{Camera}} & \multicolumn{2}{c|}{Resolution} & $4160\,\mathrm{px}\times3120\,\mathrm{px}$ \\\hline
\multicolumn{1}{|c|}{LC-VLP} & \multicolumn{2}{c|}{Sampling ratio} & $\varsigma=0.25$ \\\hline
\end{tabular}
}}
\end{table}

Fig.~\ref{fig:exp_imp}\subref{fig:imp} illustrates the hardware block diagram of the proposed positioning prototype system. On the transmitter side, the control program is developed offline on a computer and deployed on an Arduino UNO development board. The LED ID is encoded using Manchester coding and modulated via on-off keying (OOK). A dimming module together with a DD312 driver module drives the LED to flicker at a high frequency, enabling the cyclic transmission of VLC packets containing the LED ID. On the receiver side, the prototype employs the rear camera of a realme Q2i smartphone, which is a CMOS camera capable of receiving VLC signals through RSE while simultaneously capturing visual features such as LED contours. {It should be noted that LC-VLP can be implemented using any CMOS cameras after proper calibration.} Specifically, the camera consecutively captures two frames within a very short time interval using a long (e.g., $6.67\,\mathrm{ms}$) and a short exposure time (e.g., $1.25\,\mathrm{ms}$), respectively. The former is used for extracting visual features, while the latter is used for decoding VLC information to obtain the LED IDs\cite{vpca,vovlp}. {Note that by properly adjusting the short exposure time, clear VLC stripe patterns of different LEDs can still be reliably distinguished even under relatively dense LED deployments.} Owing to the short inter-frame interval, the camera pose is typically assumed to remain unchanged between the two frames\cite{vpca,vovlp}. These measurements are processed by a software application developed on the Android Studio to estimate the camera pose, with the results visualized in real time on a front-end interface such as a map view.

\subsection{Experimental Results}

\begin{figure*}[!t]
\begin{subfigure}{0.315\textwidth}\centering
    \includegraphics[width=\textwidth]{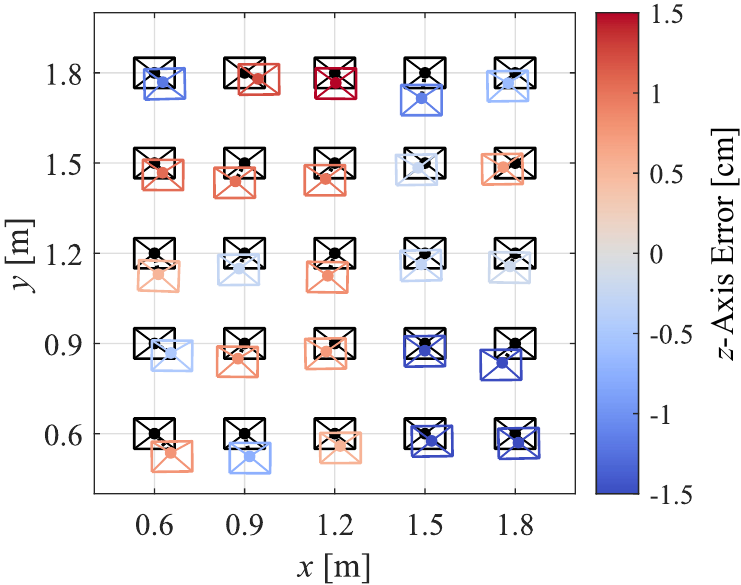}
    \caption{Height $= 1.0\m$.}\label{fig:exp_z10}
\end{subfigure}
\begin{subfigure}{0.315\textwidth}\centering
    \includegraphics[width=\textwidth]{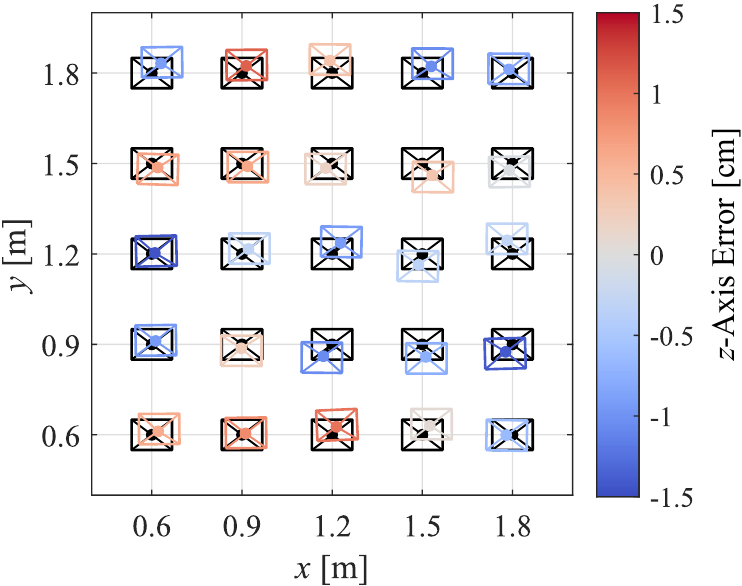}
    \caption{Height $= 1.4\m$.}\label{fig:exp_z14}
\end{subfigure}
\begin{subfigure}{0.315\textwidth}\centering
    \includegraphics[width=\textwidth]{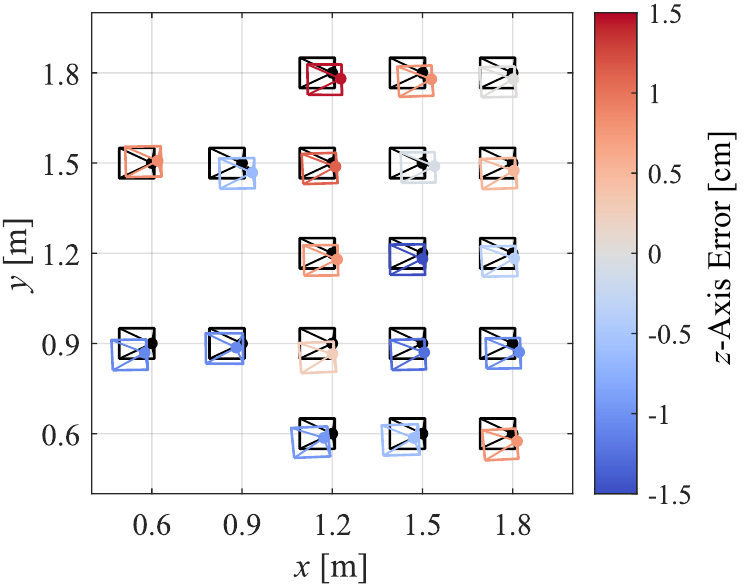}
    \caption{Rotate by $30^\circ$ about $y$-axis.}\label{fig:exp_y_30}
\end{subfigure}
\begin{subfigure}{0.315\textwidth}\centering
    \includegraphics[width=\textwidth]{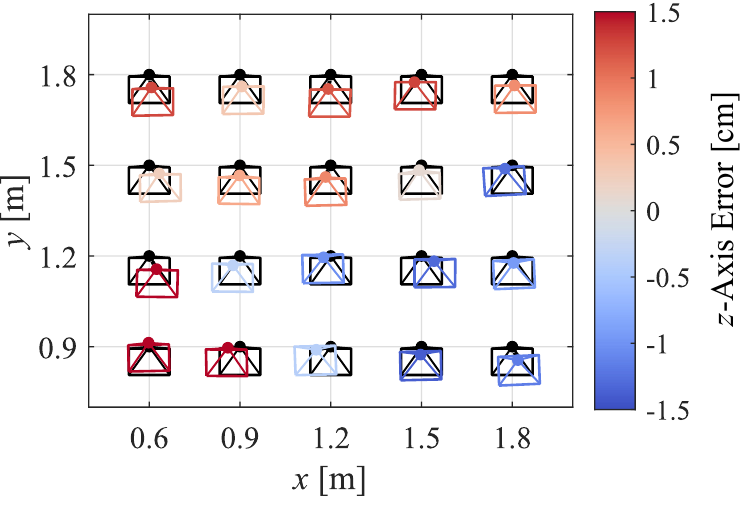}
    \caption{Rotate by $30^\circ$ about $x$-axis.}\label{fig:exp_x_30}
\end{subfigure}
\hspace{0.1cm}
\begin{subfigure}{0.265\textwidth}\centering
    \includegraphics[width=\textwidth]{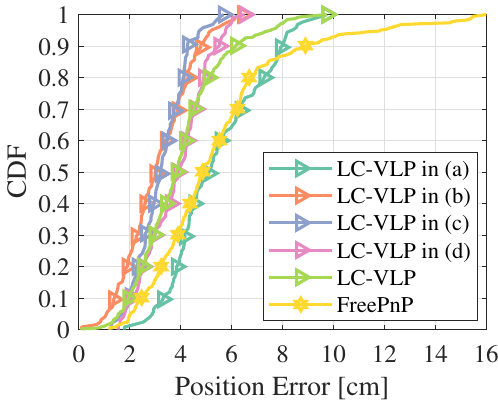}
    \caption{{CDFs of position errors.}}\label{fig:exp_cdf}
\end{subfigure}
\hspace{0.1cm}
\begin{subfigure}{0.265\textwidth}\centering
    \includegraphics[width=\textwidth]{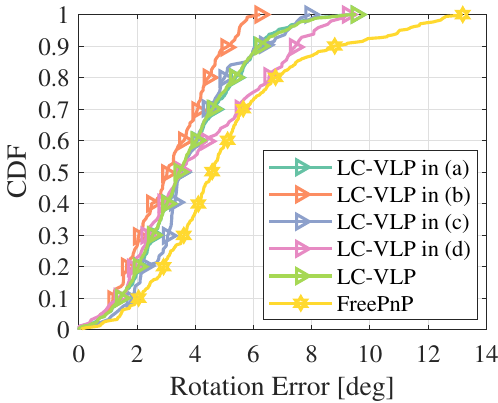}
    \caption{{CDFs of rotation errors.}}\label{fig:exp_rot_cdf}
\end{subfigure}
\centering
\caption{The experimental results of LC-VLP. (a) and (b) show the estimated camera poses (top view) when the smartphone is placed horizontally at different heights; (c) and (d) show the estimated camera poses when the smartphone, positioned at a height of $1.4\m$, is rotated about the $x$- and $y$-axes of WCS, respectively. Note that the black camera FoV cone denote the reference poses, while the colored cones represent the poses estimated by LC-VLP, where the color encodes the height difference along $z$-axis. Finally, (e) and (f) summarize the CDFs of position and rotation errors of LC-VLP and Free\pnp.}
\label{fig:exp}
\end{figure*}

To evaluate the impact of camera height and tilt angle on the performance of LC-VLP, the $2.4\m\times2.4\m$ area was discretized into 25 grid points, and the smartphone camera was placed at each grid point to capture LED images for CPE. {The camera intrinsics and distortion coefficients were calibrated in advance using a standard checkerboard-based calibration method\cite{Zhang2000}, while the extrinsics were determined using floor coordinate grids, an adjustable phone stand, and the smartphone leveling tool, with estimated measurement uncertainties within approximately $[-0.2\cm, 0.2\cm]$ for position and $[-1^\circ, 1^\circ]$ for orientation. For each grid point, we locate the user for five times.} To ensure successful positioning, the camera was required to observe at least two LEDs at all times. When the camera is tilted about the $x$- or $y$-axis, fewer than two LEDs may be visible at certain grid points; such ``unlocalizable" points are excluded from statistical analyses.

In Figs.~\ref{fig:exp}\subref{fig:exp_z10} and \subref{fig:exp_z14}, we first keep the camera horizontal and perform positioning at two different heights, i.e., $1.0\m$ and $1.4\m$, to examine the effect of camera height on the performance of LC-VLP. To more clearly visualize the full 6-DoF camera poses, the top-view camera FoV cone representation in\cite{Kang2025} is adopted. It can be observed that, compared with the results at $1.0\m$, the estimated camera poses at $1.4\m$ are noticeably closer to the reference poses. Specifically, the MPEs at heights of $1.0\m$ and $1.4\m$ are $5.36\cm$ and $3.20\cm$, respectively, while the corresponding MREs are $3.64^\circ$ and $3.11^\circ$. These results indicate that the accuracy of LC-VLP improves as the camera height increases, which is because a higher camera captures larger LED images, which is equivalent to increasing the effective LED size. Enlarging the LED size can lead to higher accuracy, which is consistent with the simulation results. {We also note that a lower camera height can increase the likelihood of observing more LEDs, which can partially mitigate the impact of image noise. However, in our experimental setup, the increase in the number of visible LEDs is limited; therefore, the LED size becomes the dominant influence factor.}

In Figs.~\ref{fig:exp}\subref{fig:exp_y_30} and \subref{fig:exp_x_30}, the camera height is fixed at $1.4\m$, while the camera is rotated by $30^\circ$ about $y$- and $x$-axes, respectively. Together with Fig.~\ref{fig:exp}\subref{fig:exp_z14}, these results are used to evaluate the impact of camera tilt on the performance of LC-VLP. It can be observed that, under different tilt angles, the estimated camera poses remain closely aligned with the reference poses. Specifically, when the camera is rotated by $30^\circ$ about $y$-and $x$-axes, the MPEs are $3.25\cm$ and $3.86\cm$, respectively, while the corresponding MREs are $3.79^\circ$ and $4.09^\circ$. These values are comparable to those obtained at the same height without rotation ($3.20\cm$ and $3.11^\circ$), indicating that camera tilt has a limited impact. 

The above results are summarized in Figs.~\ref{fig:exp}\subref{fig:exp_cdf} and \subref{fig:exp_rot_cdf}. It can be observed that, across all experimental conditions, LC-VLP achieves an overall MPE of $3.96\cm$ and an MRE of $3.65^\circ$. {We also show the results of Free\pnp\ as an ablation study, whose MPE and MRE are $5.67\cm$ and $5.16^\circ$, respectively. These results demonstrate that LC-VLP can maintain high positioning accuracy and pose estimation accuracy under varying camera heights and different camera tilt angles, while Free\pnp\ can provide a relatively reliable initialization for subsequent optimization. Moreover, the number of fully visible LEDs at each of the grid point are provided in Fig.~\ref{fig:exp_vl}. We can see that the experiments cover scenarios with both large and small numbers of visible LEDs. Therefore, this setup is able to evaluate the performance of LC-VLP under different LED visibility conditions.} 

\begin{figure*}[!t]
\begin{subfigure}{0.22\textwidth}\centering
    \includegraphics[width=\textwidth]{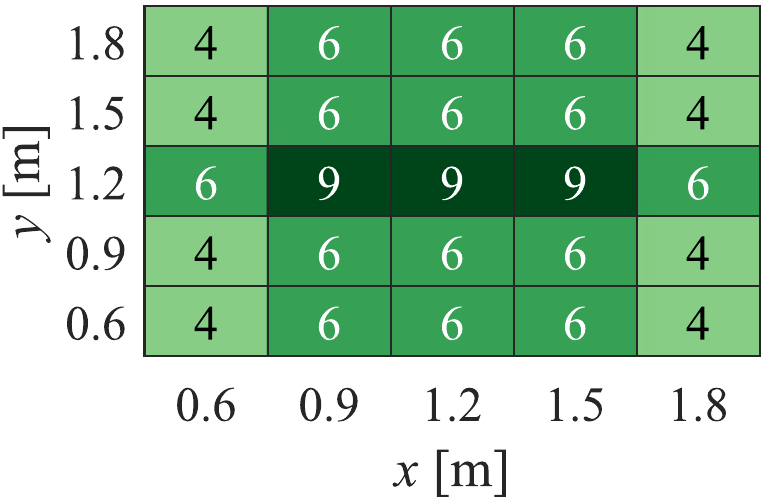}
    \caption{Height $= 1.0\m$.}\label{fig:exp_10}
\end{subfigure}
\begin{subfigure}{0.22\textwidth}\centering
    \includegraphics[width=\textwidth]{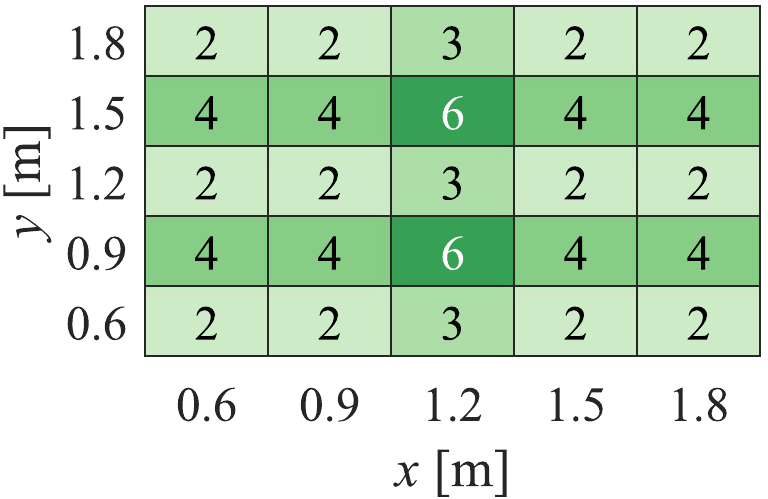}
    \caption{Height $= 1.4\m$.}\label{fig:exp_14}
\end{subfigure}
\begin{subfigure}{0.22\textwidth}\centering
    \includegraphics[width=\textwidth]{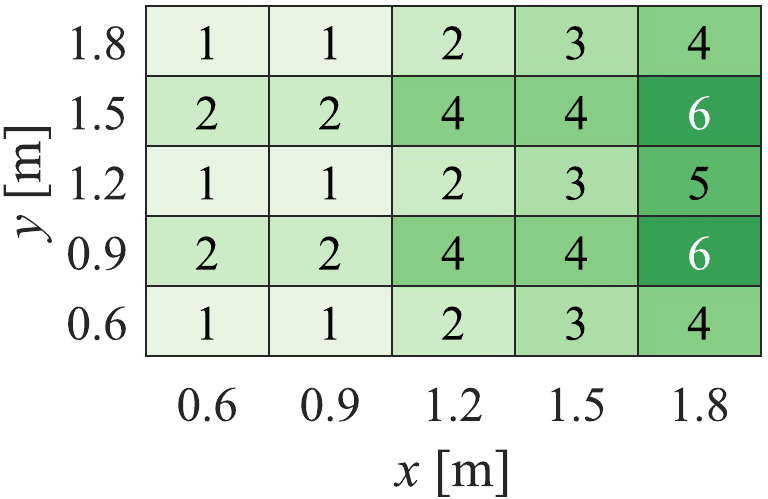}
    \caption{Rotate by $30^\circ$ about $y$-axis.}\label{fig:exp_14y30}
\end{subfigure}
\begin{subfigure}{0.255\textwidth}\centering
    \includegraphics[width=\textwidth]{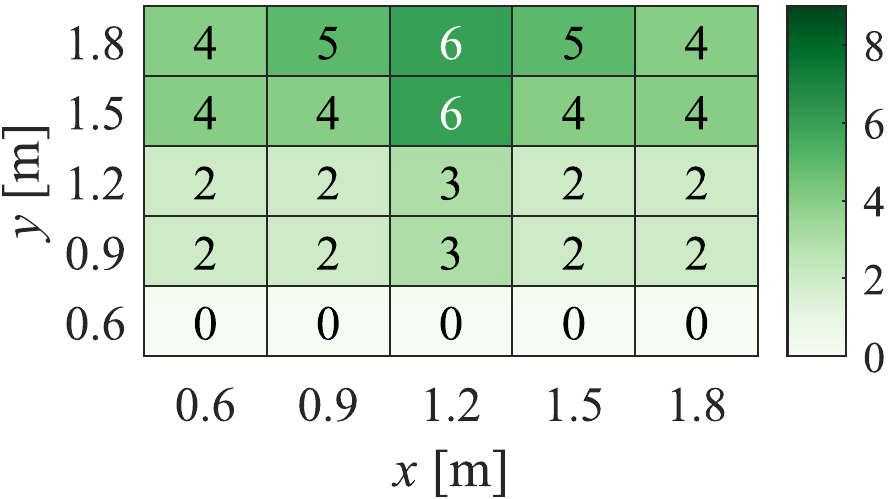}
    \caption{Rotate by $30^\circ$ about $x$-axis.}\label{fig:exp_14x30}
\end{subfigure}
\centering
\caption{{Number of fully visible LEDs in the four cases of the experiments.}}
\label{fig:exp_vl}
\end{figure*}

\begin{figure}[!t]
\begin{subfigure}{0.235\textwidth}\centering
    \includegraphics[width=\textwidth]{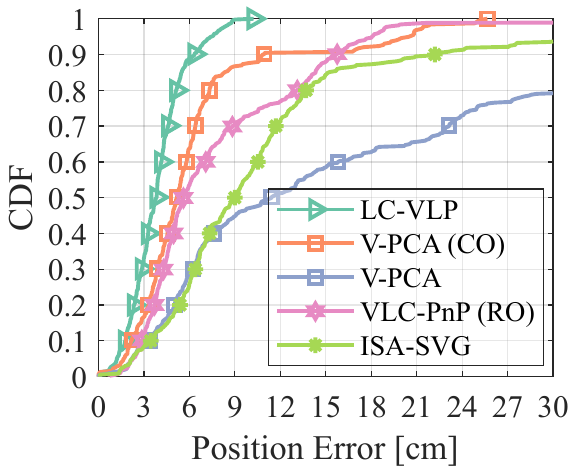}
    \caption{}\label{fig:exp_pos_bl}
\end{subfigure}
\begin{subfigure}{0.235\textwidth}\centering
    \includegraphics[width=\textwidth]{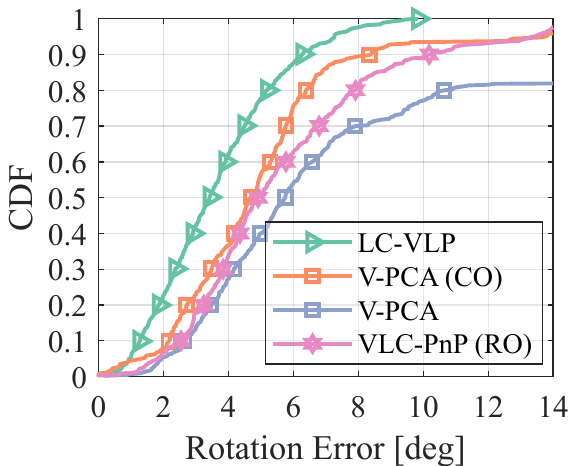}
    \caption{}\label{fig:exp_rot_bl}
\end{subfigure}
\centering
\caption{{Comparison of experimental results between LC-VLP and baseline methods.}}
\label{fig:exp_cdfs_bl}
\end{figure}

{We also conduct V-PCA, VLC-\pnp, and ISA-SVG as baselines in the four cases of the experiments. The results are presented in Fig.~\ref{fig:exp_cdfs_bl}. We note that:
\begin{itemize}
    \item V-PCA is evaluated under two settings: (i) In the first setting, all captured LEDs are treated as circular LEDs, and two of them are randomly selected for positioning; (ii) In the second setting, denoted as ``CO", the square LEDs are ignored by assuming that only the circular-LED subset exists, and the algorithm is applied only when at least two circular LEDs are captured. 
    \item VLC-\pnp\ always ignores circular LEDs and is evaluated only on the square-LED subset, denoted as ``RO".
\end{itemize}
}

{As shown in Fig.~\ref{fig:exp_cdfs_bl}\subref{fig:exp_pos_bl}, LC-VLP achieves the best positioning performance, followed by V-PCA (CO) and VLC-\pnp\ (RO), which operate under their corresponding LED-shape scenarios, with MPEs of $6.23\cm$ and $7.97\cm$, respectively. ISA-SVG ranks next, yielding an MPE of $11.34\cm$ due to the influence of IMU measurement errors. In contrast, V-PCA applied to all heterogeneous LED shapes exhibits the largest MPE, exceeding $20\cm$. Notably, the degradation of LC-VLP is less severe than that observed in simulation Scenario~\ref{sc:4}, since more than half of the LEDs in the experimental scene are still circular. As illustrated in Fig.~\ref{fig:exp_cdfs_bl}\subref{fig:exp_rot_bl}, LC-VLP also achieves the lowest rotation error, outperforming V-PCA (CO) and VLC-\pnp\ (RO), whose MREs are $5.29^\circ$ and $5.79^\circ$, respectively. These results demonstrate that LC-VLP achieves superior performance in scenarios with mixed deployments of different LED shapes.}

\begin{figure*}[!t]
\begin{subfigure}{0.235\textwidth}\centering
    \includegraphics[width=\textwidth]{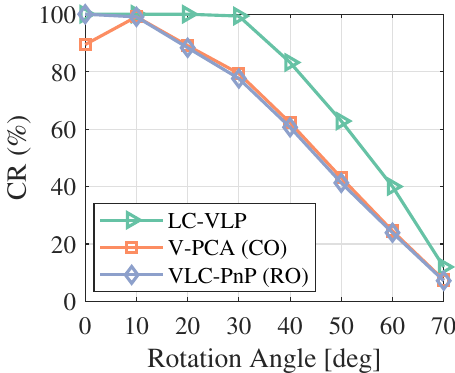}
    \caption{Height $= 1.0\m$, rotate about $y$-axis.}\label{fig:exp_cr10y}
\end{subfigure}
\begin{subfigure}{0.235\textwidth}\centering
    \includegraphics[width=\textwidth]{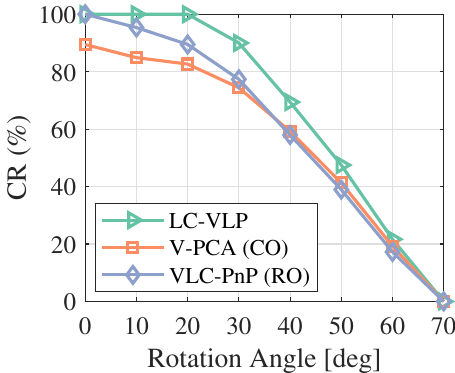}
    \caption{Height $= 1.0\m$, rotate about $x$-axis.}\label{fig:exp_cr10x}
\end{subfigure}
\begin{subfigure}{0.235\textwidth}\centering
    \includegraphics[width=\textwidth]{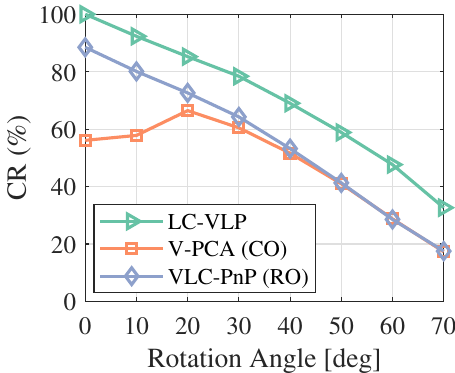}
    \caption{Height $= 1.4\m$, rotate about $y$-axis.}\label{fig:exp_cr14y}
\end{subfigure}
\begin{subfigure}{0.235\textwidth}\centering
    \includegraphics[width=\textwidth]{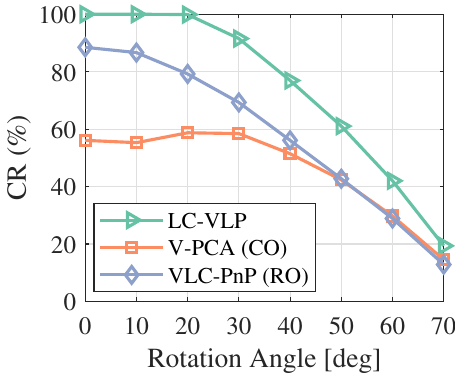}
    \caption{Height $= 1.4\m$, rotate about $x$-axis.}\label{fig:exp_cr14x}
\end{subfigure}
\centering
\caption{{CR comparison of LC-VLP with baselines under four different conditions in the experimental scenario.}}
\label{fig:exp_cr}
\end{figure*}

{We further evaluate the positioning coverage rate (CR) of LC-VLP and baselines through Monte Carlo simulations under the experimental setup within the central region of $[0.6,1.8]\times[0.6,1.8]\m^2$. The CR is defined as the probability that the camera pose can be successfully estimated when the camera, with a fixed orientation, captures LED images on a plane with a fixed height. For LC-VLP and ISA-SVG, the condition is identical, i.e., at least two LEDs of arbitrary shapes must be captured. In contrast, V-PCA (CO) and VLC-\pnp\ (RO) require at least two circular LEDs and one rectangular LED to be observed, respectively. The camera is rotated about the $x$- and $y$-axes at heights of $1.0\m$ and $1.4\m$, respectively. For each camera orientation, $10^4$ random camera positions were generated on the corresponding height plane to estimate the CR statistically. The results are shown in Figs.~\ref{fig:exp_cr}\subref{fig:exp_cr10y}--\subref{fig:exp_cr14x}. It can be observed that, under all four experimental conditions, LC-VLP consistently achieves the highest CR regardless of the camera rotation angle. VLC-\pnp\ (RO) achieves the second-highest CR since it requires fewer LEDs for positioning. These results further verify the effectiveness of LC-VLP in heterogeneous LED-shape scenarios.}

{The average CPU runtimes of LC-VLP, V-PCA, VLC-\pnp, and ISA-SVG are further evaluated on a laptop equipped with an Intel\textsuperscript{\textregistered} Core\textsuperscript{\texttrademark‌} i7-13700H processor, where all algorithms are implemented in Python 3.12. The corresponding average runtimes are $0.1488\,\mathrm{s}$, $0.0238\,\mathrm{s}$, $0.0058\,\mathrm{s}$, and $0.0009\,\mathrm{s}$, respectively. Although LC-VLP introduces additional computational overhead due to the contour-based nonlinear optimization, the proposed method still achieves practical real-time performance while offering significantly improved positioning accuracy.}

{Note that all the above experiments are conducted under normal indoor ambient lighting conditions. The indoor illumination generally does not directly affect the positioning accuracy of the proposed algorithm, since the algorithm mainly depends on the accuracy of the extracted LED contour points. In practice, the camera parameters can always be appropriately adjusted\cite{vpca,vovlp} to make the LEDs appear clearer and more distinguishable in the captured images, thereby facilitating more accurate contour extraction.}

\section{Conclusion}\label{sec:con}
In this paper, a generic LC-VLP algorithm has been proposed to enable high-accuracy full 6-DoF CPE in heterogeneous LED-shape scenarios. LC-VLP adopts a unified parametric modeling based on Lam\'e curves, which can represent a wide range of commonly-used LED shapes. With only the parameters of the captured Lam\'e curve-shaped LEDs, full 6-DoF CPE can be performed. Specifically, an LED database is constructed in advance to associate LED IDs with their corresponding curve parameters. When positioning, a back-projection strategy is introduced to formulate CPE as an NLLS optimization problem, where the algebraic distances between back-projected points and the corresponding LED curves are iteratively minimized. To provide a reliable initialization, a \freepnp\ algorithm is further developed, which enables approximate pose estimation without requiring any pre-calibrated 3D--2D RP correspondences. The performance of LC-VLP has been verified through substantial simulations and experiments. In particular, simulation results show that LC-VLP outperforms SoTA methods in homogeneous LED-shape scenarios, achieving reductions of over 30\% in average position and rotation errors; experiments further show that LC-VLP can achieve an average position accuracy of less than 4 cm in a heterogeneous LED-shape scenario. Therefore, LC-VLP holds strong potential for deployment in real-world indoor positioning-enabled applications. {In future work, it would be interesting to extend LC-VLP to stereo-camera or multi-camera cooperative positioning scenarios, investigate single-LED positioning methods, and optimize LED deployment layouts for improved coverage and positioning performance.}

\appendices
\section{Proof of Theorem~\ref{thm:ioc}} \label{app:ioc}
\begin{IEEEproof}
Let the camera optical center in WCS be denoted by $\bm{c}$. Consider three points with world coordinates $\bm{x}_k$, $k=1,2,3$, and let the pixel coordinates of their corresponding projections on the image plane be $\bm{u}_k^{\mathrm{p}}$.

\textit{(Sufficiency)}
Since the three points $\bm{x}_k$ are collinear, there exists a non-zero scalar $\lambda$ such that:
\begin{equation}
\label{eqn:w}
    \bm{x}_1-\bm{x}_2=\lambda(\bm{x}_2-\bm{x}_3).
\end{equation}
Substituting \eqref{eqn:w2c} into \eqref{eqn:w} yields:
\begin{align}
    \bm{x}_1^\mathrm{c}-\bm{x}_2^\mathrm{c} &=\bm{R}(\bm{x}_1-\bm{x}_2)+\bm{t}-\bm{t}=\lambda\bm{R}(\bm{x}_2-\bm{x}_3)\notag\\
    &=\lambda\big((\bm{R}\bm{x}_2+\bm{t})\!-\!(\bm{R}\bm{x}_3+\bm{t})\big)=\lambda(\bm{x}_2^\mathrm{c}-\bm{x}_3^\mathrm{c}),
\end{align}
which shows that the points $\bm{x}_k$ remain collinear in CCS, i.e., their camera coordinates $\bm{x}_k^{\mathrm{c}}$ are linearly dependent. Concatenating $\bm{x}_k^{\mathrm{c}}$, $k=1,2,3$ into a matrix, we have:
\begin{equation}
\label{eqn:ccc}
    \det\big[ \bm{x}_1^\mathrm{c},\bm{x}_2^\mathrm{c},\bm{x}_3^\mathrm{c} \big]=0.
\end{equation}
Let $s_k$ denote the $z$-coordinate of $\bm{x}_k^{\mathrm{c}}$. By multiplying both sides of \eqref{eqn:ccc} by $\tfrac{1}{s_1s_2s_3}$, we obtain:
\begin{equation}
    \det\big[ \tfrac{1}{s_1}\bm{x}_1^\mathrm{c}, \tfrac{1}{s_2}\bm{x}_2^\mathrm{c}, \tfrac{1}{s_3}\bm{x}_3^\mathrm{c} \big]=\tfrac{1}{s_1 s_2 s_3}\det\big[ \bm{x}_1^\mathrm{c},\bm{x}_2^\mathrm{c},\bm{x}_3^\mathrm{c} \big]=0.
\end{equation}
Furthermore, since the intrinsic matrix $\bm{K}$ is invertible, multiplying $\det(\bm{K}^{-1})$ yields:
\begin{equation}
    \det(\bm{K}^{-1}) \cdot \det\big[ \tfrac{1}{s_1}\bm{x}_1^\mathrm{c}, \tfrac{1}{s_2}\bm{x}_2^\mathrm{c}, \tfrac{1}{s_3}\bm{x}_3^\mathrm{c} \big]=0.
\end{equation}
Then from \eqref{eqn:c2p} we have:
\begin{equation}
    \det\!\big[ \widetilde{\bm{u}}_1^\mathrm{p},\widetilde{\bm{u}}_2^\mathrm{p},\widetilde{\bm{u}}_3^\mathrm{p} \big] \!=\! \det\!\Big\{\bm{K}^{-1}\big[ \tfrac{1}{s_1}\bm{x}_1^\mathrm{c},\tfrac{1}{s_2}\bm{x}_2^\mathrm{c},\tfrac{1}{s_3}\bm{x}_3^\mathrm{c} \big]\Big\}\!=\!0,
\end{equation}
which means the $\widetilde{\bm{u}}_k^\mathrm{p}$, $k=1,2,3$ are linearly dependent, i.e., there exist $\lambda_2\lambda_3\neq0$ such that $\widetilde{\bm{u}}_1^\mathrm{p}+\lambda_2\widetilde{\bm{u}}_2^\mathrm{p}+\lambda_3\widetilde{\bm{u}}_3^\mathrm{p}=\mathbf{0}$. Since $\widetilde{\bm{u}}_k^\mathrm{p}=[(\bm{u}_k^\mathrm{p})^\T,1]^\T$, we also have:
\begin{subequations}
\label{eqn:ppp}
\begin{numcases}{}
    \bm{u}_1^\mathrm{p}+\lambda_2\bm{u}_2^\mathrm{p}+\lambda_3\bm{u}_3^\mathrm{p}=\mathbf{0},\\
    1+\lambda_2+\lambda_3=0,
\end{numcases}
\end{subequations}
which clearly shows the collinearity of the projected points $\bm{u}_k^\mathrm{p}$, $k=1,2,3$ in PCS.

\textit{(Necessity)} Since the projected points $\bm{u}_k^\mathrm{p}$, $k=1,2,3$ are collinear in PCS, we have \eqref{eqn:ppp}, and we can still backtrack the above procedure up to \eqref{eqn:ccc}. \eqref{eqn:ccc} shows the linear dependency of $\bm{x}_k^\mathrm{c}$, which means there exist $\alpha_1\alpha_2\alpha_3\neq0$ such that:
\begin{equation}
\label{eqn:ccc2}
    \alpha_1\bm{x}_1^\mathrm{c}+\alpha_2\bm{x}_2^\mathrm{c}+\alpha_3\bm{x}_3^\mathrm{c}=\mathbf{0}.
\end{equation}
Substituting \eqref{eqn:w2c} into \eqref{eqn:ccc2}, we have:
\begin{equation}\label{eqn:aaa}
    \alpha_1\bm{x}_1+\alpha_2\bm{x}_2+\alpha_3\bm{x}_3=-(\alpha_1+\alpha_2+\alpha_3)\bm{R}^\T\bm{t}.
\end{equation}
According to \eqref{eqn:tar}, the world coordinates of the camera optical center are $\bm{c}=-\bm{R}^\T\bm{t}$, which can be substituted into \eqref{eqn:aaa}:
\begin{equation}
    \alpha_1\bm{x}_1+\alpha_2\bm{x}_2+\alpha_3\bm{x}_3-(\alpha_1+\alpha_2+\alpha_3)\bm{c}=\mathbf{0}.
\end{equation}
This implies that if $\sum_{k=1}^3\alpha_k\neq0$, then $\bm{x}_k$, $k=1,2,3$ and $\bm{c}$ are always coplanar. Since there exists a plane that contains $\bm{x}_k$, $k=1,2,3$, but does not contain $\bm{c}$, it should hold that $\sum_{k=1}^3\alpha_k=0$ which yields:
\begin{equation}
    \bm{x}_1-\bm{x}_3=-\tfrac{\alpha_2}{\alpha_1}(\bm{x}_2-\bm{x}_3).
\end{equation}
This clearly demonstrates the collinearity of points $\bm{x}_k$, $k=1,2,3$ in WCS. Theorem~\ref{thm:ioc} is proved.
\end{IEEEproof}

\section{Proof of Proposition~\ref{pro:samp}}\label{app:samp}
\begin{IEEEproof}
To prove Proposition~\ref{pro:samp}, it suffices to show that the arc-length differential $\mathrm{d}s$ of the projected curve is approximately proportional to the polar-angle differential $\mathrm{d}\theta$ along the original LED boundary curve. 

For simplicity, we assume that the LED center is located at the origin of WCS (equivalently, CyCS), and its $\phi=0$. We consider a point on the LED boundary curve parameterized by the polar angle $\theta$. Its coordinates in CCS are obtained from WCS via the transformation $\bm{\zeta}(\theta)$. Moreover, we define the projection function $\bm{\varpi}(\bm{x}^{\mathrm{c}})$ that maps a point in CCS onto PCS. These two mappings are respectively given by:
{\begin{subequations}
\begin{numcases}{}
    \bm{\zeta}(\theta)=\bm{R}\rho(\theta)\bm{\psi}(\theta)+\bm{t},\\
    \bm{\varpi}(\bm{x}^\mathrm{c})= \frac{f}{z^\mathrm{c}} \begin{bmatrix} x^\mathrm{c}\\y^\mathrm{c} \end{bmatrix} + \begin{bmatrix} u_0\tfrac{\delta x}{\delta u}\\v_0\tfrac{\delta y}{\delta v} \end{bmatrix}.\label{equ:pcs}
\end{numcases}
\end{subequations}
We note that the ``PCS" in \eqref{equ:pcs} shares the same scale with CCS and WCS by multiplying the physical size of pixels $\tfrac{\delta x}{\delta u}$ and $\tfrac{\delta y}{\delta v}$.} It follows that the projected LED curve can be written as $\bm{s}=\bm{\varpi}(\bm{\zeta}(\theta))$. By applying the chain rule, we directly differentiate $\bm{s}$ with respect to $\theta$, yielding:
\begin{align}
    \frac{\mathrm{d}\bm{s}}{\mathrm{d}\theta} &= \frac{\mathrm{d}\bm{\varpi}}{\mathrm{d}\bm{\zeta}} \cdot \frac{\mathrm{d}\bm{\zeta}}{\mathrm{d}\theta}
    = \frac{f}{z^\mathrm{c}} \begin{bmatrix}1&0&\tfrac{-x^\mathrm{c}}{z^\mathrm{c}}\\0&1&\tfrac{-y^\mathrm{c}}{z^\mathrm{c}}\end{bmatrix} \cdot \bm{R}\begin{bmatrix}\tfrac{\mathrm{d}\rho}{\mathrm{d}\theta}\cos\theta-\rho\sin\theta\\\tfrac{\mathrm{d}\rho}{\mathrm{d}\theta}\sin\theta+\rho\cos\theta\\0\end{bmatrix}\notag\\
    &\triangleq \frac{f}{z^\mathrm{c}} \bm{\Xi} \cdot \bm{R}\bm{\rho}.
\end{align}
Therefore, we have:
\begin{equation}\label{eqn:a_dsdt}
    \left\|\frac{\mathrm{d}\bm{s}}{\mathrm{d}\theta}\right\| \leq \frac{f}{z^\mathrm{c}} \| \bm{\Xi}\|_2 \cdot\|\bm{R}\bm{\rho}\| = \frac{f}{z^\mathrm{c}} \| \bm{\Xi}\|_2 \cdot\|\bm{\rho}\|,
\end{equation}
where $\left\|\bm{\Xi}\right\|_2 = \sqrt{\lambda^{\max}_{(\bm{\Xi}\bm{\Xi}^\T)}}$ denotes the 2-norm\cite{Nicholas1992} of $\bm{\Xi}$ and $\lambda^{\max}_{(\bm{\Xi}\bm{\Xi}^\T)}$ is the maximum eigenvalue of $\bm{\Xi}\bm{\Xi}^\T$. We first calculate: 
\begin{equation}
    \bm{\Xi}\bm{\Xi}^\T=\begin{bmatrix}
        1+\big(\tfrac{x^\mathrm{c}}{z^\mathrm{c}}\big)^2 & \tfrac{x^\mathrm{c}}{z^\mathrm{c}}\cdot\tfrac{y^\mathrm{c}}{z^\mathrm{c}} \\
        \tfrac{x^\mathrm{c}}{z^\mathrm{c}}\cdot\tfrac{y^\mathrm{c}}{z^\mathrm{c}} & 1+\big(\tfrac{y^\mathrm{c}}{z^\mathrm{c}}\big)^2
    \end{bmatrix}.
\end{equation}
Then, let $\det(\bm{\Xi}\bm{\Xi}^\T-\lambda\bm{I})=0$ and we obtain the eigenvalues of $\bm{\Xi}\bm{\Xi}^\T$ by solving:
\begin{equation}\label{eqn:eigen_eq}
    \lambda^2-\Big(2+\big(\tfrac{x^\mathrm{c}}{z^\mathrm{c}}\big)^2+\big(\tfrac{y^\mathrm{c}}{z^\mathrm{c}}\big)^2\Big)\lambda+\Big(1+\big(\tfrac{x^\mathrm{c}}{z^\mathrm{c}}\big)^2+\big(\tfrac{y^\mathrm{c}}{z^\mathrm{c}}\big)^2\Big)=0,
\end{equation}
of which the larger root $\lambda=1+\big(\tfrac{x^\mathrm{c}}{z^\mathrm{c}}\big)^2+\big(\tfrac{y^\mathrm{c}}{z^\mathrm{c}}\big)^2$ is taken as the maximum eigenvalue $\lambda^{\max}_{(\bm{\Xi}\bm{\Xi}^\T)}$. Therefore, the $\left\|\bm{\Xi}\right\|_2$ is calculated as:
\begin{equation}\label{eqn:a_Cau}
    \left\|\bm{\Xi}\right\|_2 = \sqrt{\lambda^{\max}_{(\bm{\Xi}\bm{\Xi}^\T)}} = \sqrt{1+\big(\tfrac{x^\mathrm{c}}{z^\mathrm{c}}\big)^2+\big(\tfrac{y^\mathrm{c}}{z^\mathrm{c}}\big)^2}.
\end{equation}

Furthermore, since the LED is captured by the camera, it must lie within the camera's FoV. This leads to the following constraints:
\begin{subequations}
\begin{numcases}{}
    0\leq\frac{|x^\mathrm{c}|}{z^\mathrm{c}}\leq\tan\alpha_y,\\
    0\leq\frac{|y^\mathrm{c}|}{z^\mathrm{c}}\leq\tan\alpha_x,
\end{numcases}
\end{subequations}
which can then be substituted into \eqref{eqn:a_Cau}, resulting in:
\begin{equation}\label{eqn:a_dpdz}
    \|\bm{\Xi}\|_2
    \leq \sqrt{1+\tan^2\alpha_y+\tan^2\alpha_x} \triangleq \Omega_{\text{FoV}}.
\end{equation}
Substituting \eqref{eqn:a_dpdz} into \eqref{eqn:a_dsdt}, we obtain that the objective ratio $\frac{\mathrm{d}s}{\mathrm{d}\theta}$ satisfies:
\begin{align}\label{eqn:a_dsdt2}
    \frac{\mathrm{d}s}{\mathrm{d}\theta} &= \left\|\frac{\mathrm{d}\bm{s}}{\mathrm{d}\theta}\right\| \leq \frac{f}{z^\mathrm{c}}\Omega_{\text{FoV}}\|\bm{\rho}\| \notag\\
    &= \frac{f}{z^\mathrm{c}}\Omega_{\text{FoV}} \sqrt{\left(\frac{\mathrm{d}\rho}{\mathrm{d}\theta}\right)^2+\rho^2} = \frac{f}{z^\mathrm{c}}\Omega_{\text{FoV}} \frac{\mathrm{d}\ell}{\mathrm{d}\theta}.
\end{align}
We can observe that $\frac{\mathrm{d}\ell}{\mathrm{d}\theta}$ corresponds exactly to the arc-length differential of the original LED curve in WCS, which naturally satisfies:
\begin{equation}
    \int_{0}^{2\pi}\frac{\mathrm{d}\ell}{\mathrm{d}\theta}\mathrm{d}\theta=L,\quad\frac{\mathrm{d}\ell}{\mathrm{d}\theta}>0,
\end{equation}
where $L$ denotes the perimeter of the LED. This implies that the arc-length differential $\frac{\mathrm{d}\ell}{\mathrm{d}\theta}$ is bounded, nonnegative, and varies smoothly with respect to $\theta$. 

{Moreover, in typical indoor positioning scenarios, it holds that $f\Omega_{\text{FoV}}\ll z^\mathrm{c}$, since $f$ is typically at the millimeter scale; constant $\Omega_{\text{FoV}}$ is on the same order of magnitude as $1$; whereas $z^\mathrm{c}$ is usually at the meter scale with sufficient camera-to-LED distance.} Substituting this condition into \eqref{eqn:a_dsdt2}, we observe that $\frac{\mathrm{d}s}{\mathrm{d}\theta}$ also varies slowly with $\theta$. Consequently, the projected arc-length differential $\mathrm{d}s$ is approximately proportional to the polar-angle differential $\mathrm{d}\theta$. 
Proposition~\ref{pro:samp} is proved.
\end{IEEEproof}

\bibliography{references}
\bibliographystyle{IEEEtran}

\end{document}